\documentclass[12pt]{article}
\usepackage[margin=1.1in]{geometry}
\usepackage{setspace}
\usepackage{amsmath}
\usepackage{amssymb}
\usepackage{amsthm}
\usepackage{forloop}
\usepackage{nicefrac}
\usepackage{mathtools}
\usepackage[table]{xcolor}
\usepackage{url}
\usepackage{bm}
\usepackage{graphicx}
\usepackage{float}
\usepackage{soul}
\usepackage{hyperref}
\hypersetup{
    colorlinks,
    linkcolor={blue!100!black!70},
    citecolor={blue!75!black},
    urlcolor={blue!75!black}
}
\usepackage[nameinlink,capitalise]{cleveref}
\usepackage[shortlabels]{enumitem}
\usepackage{cancel}
\usepackage{multirow}
\usepackage{makecell}
\usepackage{caption}
\usepackage{subcaption}
\usepackage{framed}
\usepackage{wrapfig}
\usepackage{booktabs}
\usepackage{mathrsfs}
\usepackage[linesnumbered, ruled]{algorithm2e}
 \usepackage{subfiles}
\usepackage[T1]{fontenc}
\usepackage[utf8]{inputenc} 
\usepackage{tikz}
\usetikzlibrary{arrows.meta,positioning,shapes,calc}
\usepackage{xcolor}
\usepackage{pifont} 
\usepackage{float}
\usepackage{subcaption}
\usepackage{booktabs}
\usepackage{multirow}
\usepackage{adjustbox}
\usepackage{pdflscape}
\usepackage{threeparttable} 

\usepackage[round]{natbib}
\usepackage[splitrule]{footmisc}
\usepackage{makecell}
\usepackage{amsthm} %

\usepackage{graphicx}
\usepackage{booktabs}
\usepackage{threeparttable}
\usepackage{multirow}
\usepackage{siunitx}
\usepackage{pdflscape}

\definecolor{goodcolor}{RGB}{34,34,34}    
\definecolor{badcolor}{RGB}{120,120,120}  
\definecolor{boxcolor}{RGB}{230,230,230}  
\usepackage{amsmath}
\tikzset{
  var/.style = {
    rectangle, draw=goodcolor, very thick, rounded corners,
    align=center, minimum width=2.2cm, minimum height=1cm, fill=boxcolor
  },
  arrowgood/.style = {->, >=Latex, line width=1pt, draw=goodcolor},
  arrowbad/.style  = {->, >=Latex, line width=1pt, draw=badcolor, dashed},
  note/.style = {font=\footnotesize, align=center}
}

\theoremstyle{assumpstyle}
\newtheorem{assumption}{Assumption}

\newtheorem{theorem}{Theorem}[section]

\newtheorem{corollary}{Corollary}[section]

\newtheorem{proposition}{Proposition}[section]

\crefname{definition}{Def.}{Defs.}
\crefname{equation}{Eq.}{Eqs.}

\newcommand{\defcal} [1]{\expandafter\newcommand\csname cal#1\endcsname{{\mathcal #1}}}
\newcommand{\defsf} [1]{\expandafter\newcommand\csname sf#1\endcsname{{\mathsf #1}}}
\newcommand{\defbf} [1]{\expandafter\newcommand\csname bf#1\endcsname{{\mathbf #1}}}
\newcommand{\defbb} [1]{\expandafter\newcommand\csname bb#1\endcsname{{\mathbb{#1}}}}
\newcommand{\deffrak} [1]{\expandafter\newcommand\csname frak#1\endcsname{{\mathfrak{#1}}}}

\newcounter{ct}
\forLoop{1}{26}{ct}{
    \edef\letter{\Alph{ct}}
    \expandafter\defcal\letter
    \expandafter\defsf\letter
    \expandafter\defbf\letter
    \expandafter\defbb\letter
    \expandafter\deffrak\letter
}
\forLoop{1}{26}{ct}{
    \edef\letter{\alph{ct}}
    \expandafter\defcal\letter
    \expandafter\defsf\letter
    \expandafter\defbf\letter
    \expandafter\defbb\letter
    \expandafter\deffrak\letter
}

\newcommand{\mirrorfunc}[1]{
    \ensuremath{\if$#1$\phi \else \phi(#1)\fi}
}
\newcommand{\mirrorfuncdual}[1]{
    \ensuremath{\if$#1$\phi^{*}\else \phi^{*}(#1)\fi}
}

\newcommand{\subscript}[2]{$#1 _ #2$}
\newlist{assumplist}{enumerate}{1}
\setlist[assumplist]{label=\subscript{\textbf{\textsf{A}}}{\textsf{{\arabic*}}},leftmargin=*, itemsep=0pt}

\newcommand{\potential}[1]{
    \ensuremath{\if$#1$f\else f(#1)\fi}
}
\newcommand{\metric}[1]{
    \ensuremath{\if$#1$\mathscr{G}\else \mathscr{G}(#1)\fi}
}

\makeatletter
\newcommand*{\algotitle}[2]{%
  \stepcounter{algocf}%
  \hypertarget{algocf.title.\theHalgocf}{}%
  \NR@gettitle{#1}%
  \label{#2}%
  \addtocounter{algocf}{-1}%
}
\makeatother

\allowdisplaybreaks

\crefname{paragraph}{part}{parts}
\makeatletter
\def\@maketitle{%
  \newpage
  \begin{center}%
  \let \footnote \thanks
    {\Large \bf \@title  \par}
  \end{center}%
  \par
  \vskip 0.5em
  }

\makeatother

\title{Adaptive Penalization and Bootstrap-Smoothed Inference for Two-Sample Mendelian Randomization with Summary Data}

\begin{document}
\maketitle

\begin{center}
{\large
\begin{tabular}{ccc}
    \makecell{Muhammad Qasim\(^\dagger\)\(^\star\)\\{\normalsize\texttt{muhammad.qasim@stat.lu.se}}} & \makecell{Kai Wang\(^\ddagger\)\\{\normalsize\texttt{kai-wang@uiowa.edu}}} & \makecell{Ishan S. Bhatt\(^\ast\)\\{\normalsize\texttt{ishan.bhatt@okstate.edu}}} 
\end{tabular}
\vskip 0.5em

\normalsize
\begin{tabular}{c}
\({}^{\dagger}\)Department of Satistics, Lund University, Sweden\\
\({}^{\dagger}\)Department of Biostatistics, College of Public Health, University of Iowa, USA\\
\({}^{\ast}\)Department of Communication Sciences and Disorders, Oklahoma State University, USA\\

\({}^{\star}\)\textit{Corresponding author:} {\normalsize\texttt{muhammad.qasim@stat.lu.se}}\\
\end{tabular}\\
\vskip 1.1em
}
\end{center}
\vskip 0.5em

\begin{abstract}
\singlespacing
   Two-sample Mendelian randomization (MR) uses genetic variants as instrumental variables to estimate causal effects from observational data using summary association statistics. However, horizontal pleiotropy can invalidate standard MR estimators and lead to biased causal inference.  Pleiotropy-robust methods have been proposed to address this issue, including regularization-based approaches such as MR-Lasso. However, MR-Lasso may fail to identify invalid instruments consistently, and its post-selection inference can be unreliable. In this paper, we develop two lasso-type procedures for two-sample MR with summary-level data. The first, MR-ALasso, extends MR-Lasso by introducing adaptive penalty weights for pleiotropic effects in order to improve the identification of valid and invalid instruments. The second, MR-ALasso-B, combines adaptive lasso selection with bootstrap smoothing to improve post-selection inference. We establish theoretical results for MR-ALasso under the two-sample summary data framework, including invalid instrument identification consistency and oracle-type post-selection behavior. Simulation studies show that MR-ALasso generally improves upon MR-Lasso in estimation accuracy and invalid-instrument identification, whereas MR-ALasso-B substantially improves coverage and type-I error control relative to naive post-selection inference. A real-data application based on bidirectional analyses of multiple complex traits further illustrates the practical usefulness of the proposed methods. We provide an \textsf{R} package, \texttt{MRAlasso}, to facilitate implementation. 
\end{abstract}

\setcounter{page}{1}

\section{Introduction}
Mendelian randomization (MR) is widely used to investigate causal relationships between exposures and outcomes using genetic variants, commonly single-nucleotide polymorphisms (SNPs), as instrumental variables \citep{DaveySmith2003,lawlor2008mendelian}. With the increasing availability of summary association statistics from large-scale genome-wide association studies (GWASs), two-sample MR has become a standard design in genetic epidemiology. In this setting, SNP--exposure associations, such as SNP--body mass index (BMI) associations, and SNP--outcome associations, such as SNP--type 2 diabetes (T2D) associations, can be obtained from separate samples without requiring access to individual-level data \citep{burgess2013mendelian,bowden2019improving}. For a genetic variant to be a valid instrumental variable, three core assumptions should hold. First, the variant must be associated with the exposure of interest (relevance). Second, it should be independent of unmeasured confounders of the exposure--outcome relationship (independence). Third, it must affect the outcome only through the exposure (exclusion restriction). When these assumptions hold for all selected variants, a standard estimator such as inverse variance weighting (IVW) is consistent for estimating the causal effect \citep{grant2021pleiotropy}.

In practice, however, the instrumental variable assumptions may be violated, most notably through horizontal pleiotropy, wherein a genetic variant has a direct impact on the outcome through pathways other than the exposure variable. The inclusion of a variant in an MR analysis that violates either the independence or the exclusion restriction may lead to biased causal estimates \citep{burgess2013use}. Robust methods have therefore been developed to obtain consistent causal effect estimates under weaker assumptions in the presence of multiple genetic variants. These include regression-based, median-based, mode-based, mixture-based, and penalization-based estimators, among others \citep{Bowden2015, bowden2016consistent, hartwig2017robust, verbanck2018detection, rees2019robust, zhao2020statistical,xue2021constrained,slob2020comparison}. 

Much recent work has focused on developing robust estimators for MR analyses with invalid instruments. One major line of work relaxes the all-valid-instruments assumption through alternative identification conditions. \cite{bowden2016consistent} proposed the weighted median estimator, which is a consistent estimate of the causal effect provided that at least half of the instruments are valid. Similarly, the mode-based estimator \citep{Hartwig2017} is motivated by the observation that, if the largest cluster of ratio estimates corresponds to valid instruments, the mode of the distribution of ratio estimates provides a consistent estimate of the causal effect. However, its performance can be sensitive to the bandwidth used to estimate the mode and may be less efficient than IVW when all instruments are valid.

Other approaches explicitly detect and remove outlying instruments. The MR-Presso procedure \citep{verbanck2018detection} identifies variants that generate unusually large residuals in the IVW regression and removes them before re-estimating the causal effect. This method can effectively correct for pleiotropy driven by a small number of outliers, although its performance may worsen when pleiotropic effects are extensive. Robust regression approaches have also been proposed to reduce the influence of instruments with heterogeneous effects. For example, \citet{rees2019robust} introduced the MR-Robust method, which applies robust regression techniques to reduce the influence of variants with large residuals in the regression model. These methods are less sensitive to outliers than IVW and can provide more stable causal estimates when some instruments exhibit heterogeneous effects.

Another line of research focuses on more explicit statistical modelling of pleiotropy using likelihood-based, mixture-based, and regression-based frameworks. \citet{zhao2020statistical} proposed a robust adjusted profile score method to accommodate both systematic and idiosyncratic pleiotropy in two-sample summary data. \citet{burgess2020robust} proposed mixture-based approaches that aim to obtain robust causal estimates in the presence of invalid instruments. These methods identify groups of genetic variants with similar causal estimates, which may indicate shared underlying mechanisms or clusters of valid instruments around the true causal effect, while treating other variants as invalid. \citet{bowden2018improving} introduced radial plots and radial regression as a reparameterization of standard IVW and MR-Egger estimators, providing a regression-based framework. These methods can be powerful when their structural assumptions are approximately correct.

Among the available robust methods, regularization-based approaches are useful when invalid instruments are sparse. In such settings, horizontal pleiotropic effects may be regarded as nuisance parameters with an approximately sparse structure, so that variable-selection techniques can be used to identify invalid instruments while estimating the causal parameter. The application of the least absolute shrinkage and selection operator (Lasso) to instrumental variables models has been discussed in the literature for individual-level data \citep{cheng2015select,kang2016instrumental,windmeijer2019use,qasim2025lasso} for selecting valid and relevant instruments. In the summary-data MR setting, however, regularization-based methods are much less developed. An important exception is the MR-Lasso procedure of \citet{rees2019robust}, which proposed the MR-Lasso estimator with potentially invalid instruments for two-sample summary data. However, lasso-based procedures may fail to correctly identify invalid instruments when instrument strengths are heterogeneous or when the irrepresentable-type conditions required for exact recovery are not satisfied. For the individual-level data, \citet{windmeijer2019use} addressed this issue using an adaptive Lasso procedure based on a median initial estimator. Since MR analyses are often based on GWAS summary statistics rather than individual-level data, that framework cannot be transferred directly to the summary data MR.

In this paper, we extend the adaptive Lasso methodology to the two-sample summary data. We propose an adaptive MR-Lasso estimator, MR-ALasso, which modifies the MR-Lasso procedure by introducing data-dependent penalty weights for the horizontal pleiotropic effects. The aim is to improve robust causal effect estimation in the presence of horizontal pleiotropy by improving the selection of invalid instruments when pleiotropic effects are sparse and instrument strengths are heterogeneous. We study the theoretical properties of the proposed estimator and establish conditions under which invalid-instrument identification consistency and oracle-type post-selection behavior are obtained under the two-sample summary-data model. Because inference based on a single selected model may remain anti-conservative after penalized selection, we also propose a bootstrap-smoothed extension, MR-ALasso-B. This procedure averages post-selection estimators across SNP-level bootstrap resamples in order to incorporate the additional variability due to the instrument selection step. The idea is in line with a broader post-selection inference literature, in which bootstrap smoothing has been used to improve interval estimation and hypothesis testing after data-adaptive model selection \citep{efron2014estimation,wang2014discussion}. The two proposed methods, MR-ALasso and MR-ALasso-B, address two related but distinct objectives, with MR-ALasso aimed at improved estimation and invalid instrument identification, and MR-ALasso-B aimed at improved post-selection inference. We further illustrate the practical behavior of the proposed methods through a broad real-data application based on bidirectional pairwise MR analyses across multiple complex traits, showing that MR-ALasso can identify additional potentially invalid instruments and that MR-ALasso-B can provide more stable post-selection inference.

The structure of the paper is as follows. Section \ref{sec:framework} presents the two-sample MR, introduces the summary-data model, and reviews horizontal pleiotropy with existing robust estimators. Section \ref{sec:method} develops the proposed adaptive penalized MR approaches, the choice of tuning parameter, and a selection consistency result. Section \ref{sec:simulation} reports simulation studies under two data-generating models. Section \ref{sec:application} illustrates the proposed methods in a real-data application. Section \ref{sec:discussion} concludes with a discussion. Proofs of the theoretical results and additional simulation results are provided in the Supplementary Material.

\section{Two-Sample Mendelian Randomization Framework} \label{sec:framework}

Let $Z_{ij}$ denote the genotype score of individual $i$ at genetic variant $j$, for $i=1,\dots,n$ and $j=1,\dots,J$. These genetic variants, typically SNPs, are used as candidate instrumental variables. Let $X_i$ denote the exposure and $Y_i$ the outcome of interest. We consider the following linear structural model for the exposure and outcome, similar to the framework described in \citet{bowden2017framework} and \citet{kang2016instrumental} for the single–exposure case,

\begin{align}
X_i &= \sum_{j=1}^{J} \beta_{Xj} Z_{ij} + \gamma_U U_i + \varepsilon_{Xi}, \label{eq:exposure} \\
Y_i &= \theta X_i + \sum_{j=1}^{J} \alpha_j Z_{ij} + \delta_U U_i + \varepsilon_{Yi}, \label{eq:outcome}
\end{align}
where $U_i$ denotes an unobserved confounder affecting both the exposure and the outcome, and $\gamma_U$ and $\delta_U$ represent its effects on the exposure and the outcome, respectively. The error terms $\varepsilon_{Xi}$ and $\varepsilon_{Yi}$ are assumed to have mean zero and to be independent of the genetic variants. The parameter of interest $\theta$ represents the causal effect of the exposure on the outcome. The coefficients $\beta_{Xj}$ represent the association between genetic variant $j$ and the exposure, while $\alpha_j$ indicates the direct effect of variant $j$ on the outcome that is not mediated through the exposure. A genetic variant with $\alpha_j = 0$ is treated as a valid instrument, whereas a variant with $\alpha_j \neq 0$ is invalid due to horizontal pleiotropy. The directed acyclic graph corresponding to the structural model in \eqref{eq:exposure}–\eqref{eq:outcome} is shown in Figure~\ref{fig:mr_dag}.
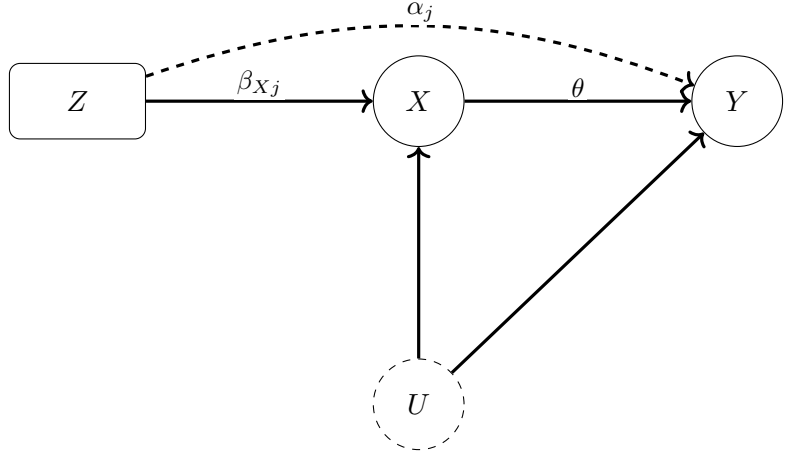
\begin{figure}[tbp]
\centering
\begin{tikzpicture}[
    node distance = 2.8cm and 3cm,
    every node/.style = {font=\small},
    snp/.style = {draw, rounded corners, minimum width=1.8cm, minimum height=1cm, align=center},
    var/.style = {draw, circle, minimum size=1.2cm, align=center},
    conf/.style = {draw, circle, dashed, minimum size=1.2cm},
    good/.style = {->, very thick},
    pleio/.style = {->, very thick, dashed},
    lab/.style = {font=\footnotesize, fill=white, inner sep=1pt}
]
\node[snp] (Z) {$Z$};
\node[var, right=of Z] (X) {$X$};
\node[var, right=of X] (Y) {$Y$};
\node[conf, below=of X] (U) {$U$};
\draw[good] (Z) -- node[lab, above] {$\beta_{Xj}$} (X);
\draw[good] (X) -- node[lab, above] {$\theta$} (Y);
\draw[good] (U) -- (X);
\draw[good] (U) -- (Y);
\draw[pleio, bend left=20] (Z) to node[lab, above] {$\alpha_j$} (Y);
\end{tikzpicture}
\caption{
Directed acyclic graph for a single genetic variant in the MR framework. Here $Z$ denotes a generic genetic variant and $U$ denotes an unobserved confounder.  The dashed edge from $Z$ and $Y$ represents possible horizontal pleiotropy.}
\label{fig:mr_dag}
\end{figure}
Identification of the causal effect $\theta$ relies on the standard instrumental variable assumptions that are commonly used in MR studies \citep{DaveySmith2003,lawlor2008mendelian}. First, the genetic variants used as instruments must be associated with the exposure, which requires that at least one of the coefficients $\beta_{Xj}$ is nonzero (relevance). Second, the genetic variants should be independent of unobserved confounders affecting both the exposure and the outcome, i.e., $Z_{ij} \perp U_i$ (independence). Third, a valid genetic variant must affect the outcome only through the exposure, which in model \eqref{eq:outcome} corresponds to $\alpha_j = 0$ (exclusion restriction). Violations of the independence or exclusion restriction assumptions can lead to biased causal estimates and invalid statistical inference \citep{burgess2013use,rees2019robust}. In practice, horizontal pleiotropy, where genetic variants influence the outcome through pathways other than the exposure, is a major concern in MR analyses. For further discussion of these assumptions and their implications, see \citet{burgess2023guidelines} and \citet{sanderson2022mendelian}.

\subsection{Summary-level data model} \label{subsec:summary_model}

In many MR applications, individual-level data are unavailable and analyses rely on summary statistics obtained from GWAS. In the two-sample MR design, the associations between genetic variants and the exposure and outcome are estimated in two independent samples; $n_X$ independent observations of $(X; Z_1,\dots,Z_J)$ from the exposure GWAS and $n_Y$ independent observations of $(Y; Z_1,\dots,Z_J)$ from the outcome GWAS. Let $\hat{\beta}_{Xj}$ and $\mathrm{SE}(\hat{\beta}_{Xj})$ denote the least squares estimate and corresponding standard error from a linear regression of $X$ on $Z_j$, respectively. Similarly, let $\hat{\beta}_{Yj}$ and $\mathrm{SE}(\hat{\beta}_{Yj})$ denote the least squares estimate and corresponding standard error from a linear regression of $Y$ on $Z_j$, respectively.
Under the structural model in \eqref{eq:exposure}--\eqref{eq:outcome}, the population-level association between SNP $j$ and the outcome satisfies
\begin{equation}
\beta_{Yj} = \theta \beta_{Xj} + \alpha_j,
\label{eq:summary_level_model}
\end{equation}
where $\alpha_j$ represents the direct pleiotropic effect of SNP $j$ on the outcome. 

\begin{assumption}
\label{assump:two-sample_summary_data}
We observe the GWAS summary associations $\hat{\beta}_{Xj}$ and $\hat{\beta}_{Yj}$ for each $j$ satisfying $\hat{\beta}_{Xj} \sim N(\beta_{Xj}, \sigma_{Xj}^2)$, and  $\hat{\beta}_{Yj} \sim N(\beta_{Yj}, \sigma_{Yj}^2),$ where $\sigma_{Xj}^2$ and $\sigma_{Yj}^2$ denote the sampling variances, estimated by $\mathrm{SE}(\hat{\beta}_{Xj})^2$ and $\mathrm{SE}(\hat{\beta}_{Yj})^2$, respectively. If the exposure and outcome GWASs are conducted in non-overlapping samples, then $\hat{\beta}_{Xj}$ and $\hat{\beta}_{Yj}$ are independent. The collection $\{(\hat{\beta}_{Xj},\hat{\beta}_{Yj})\}_{j=1}^J$ may be treated as independent after linkage disequilibrium pruning. As the GWAS sample sizes $(n_X,n_Y) \to \infty$, we assume $n_X/n_Y \to c$ for some constant $0 < c < \infty$.
\end{assumption}

Assumption \ref{assump:two-sample_summary_data} describes the standard two-sample summary data framework used in MR (e.g., \cite{zhao2020statistical,patel2024selecting}). The normal approximation is justified by the large-sample behavior of GWAS regression estimators, given the large sample sizes typically available in GWAS. When the exposure and outcome GWAS are conducted in non-overlapping samples, $\hat{\beta}_{Xj}$ and $\hat{\beta}_{Yj}$ are independent for each SNP $j$.  After linkage disequilibrium pruning, the elements of the collection $\{(\hat{\beta}_{Xj},\hat{\beta}_{Yj})\}_{j=1}^J$ are commonly treated as independent as a working approximation. Although this cross-SNP independence is not exact, \citep{zhao2020statistical} argues that cross-SNP correlations in marginal summary estimates are typically negligible because each individual variant explains only a very small proportion of phenotypic variation.

A commonly used estimator of the causal effect $\theta$ in \eqref{eq:summary_level_model} is the IVW estimator \citep{burgess2013mendelian}. The estimator can be obtained by fitting the weighted least-squares regression model
\begin{equation}
\hat{\beta}_{Yj} = \theta \hat{\beta}_{Xj} + e_j, 
\qquad e_j \sim N(0, 1/\tilde{w}_j),
\label{eq:linearmodel}
\end{equation}
where $\tilde{w}_j = 1/\sigma_{Yj}^2$ denotes the inverse-variance weight for instrument $j$. The IVW estimator is equivalently obtained by solving
$$\hat{\theta}_{\mathrm{IVW}}= \underset{\theta}{\text{argmin}} \sum_{j=1}^{J}\tilde{w}_j(\hat{\beta}_{Yj} - \theta \hat{\beta}_{Xj})^2.$$
When all genetic variants satisfy the instrumental variable assumptions so that $\alpha_j = 0$ $\forall j$, the IVW estimator is consistent for the causal effect $\theta$ \citep{grant2021pleiotropy}.

\subsection{Robust MR methods for pleiotropy}
\label{subsec:robust_methods}
In practice, some genetic variants may have direct effects on the outcome ($\alpha_j \neq 0$), a phenomenon known as horizontal pleiotropy. Balanced pleiotropy occurs when the pleiotropic effects are centered around zero, that is $E[\alpha_j]=0$, whereas directional pleiotropy occurs when $E[\alpha_j]\neq0$. The presence of pleiotropy can lead to bias in the IVW estimator and has motivated the development of different robust MR methods. Many robust MR methods rely on the instrument strength independent of direct effect (InSIDE) assumption \citep{Bowden2015}. Under this condition, the pleiotropic effects are assumed to be independent of the instrument strengths $\beta_{Xj}$. When the pleiotropic effects have mean zero and are independent of the instrument strengths, the IVW estimator remains consistent under a multiplicative random-effects model \citep{Bowden2015, rees2019robust, slob2020comparison}.

Allowing an intercept term in the regression model \eqref{eq:linearmodel} leads to the MR-Egger estimator \citep{Bowden2015},
\begin{equation}
\hat{\beta}_{Yj} = \theta_0 + \theta \hat{\beta}_{Xj} + e_j,
\label{eq:mr_egger}    
\end{equation}

where the intercept parameter $\theta_0$ represents the average pleiotropic effect across instruments. Under the InSIDE assumption, and when the SNP--exposure associations are estimated sufficiently precisely, the slope parameter $\theta$ provides a consistent estimate of the causal effect even when all instruments exhibit pleiotropy.

Distribution-based robust estimators can be described in terms of the variant-specific ratio estimates $\hat{\theta}_j = {\hat{\beta}_{Yj}}/{\hat{\beta}_{Xj}}.$ Under the summary-level model \eqref{eq:summary_level_model}, and ignoring sampling error for interpretation, these ratio estimates satisfy $\hat{\theta}_j \approx \theta + {\alpha_j}/{\beta_{Xj}}.$ Thus, for valid instruments with $\alpha_j=0$, the ratio estimates are centered around the causal effect $\theta$, but invalid instruments may yield shifted ratio estimates. The weighted median estimator is defined as the weighted median of $\{\hat{\theta}_j\}_{j=1}^J$, with weights typically based on the inverse variance of the ratio estimates. The weighted median estimator is consistent under the assumption that at least 50\% of the total weight is contributed by valid instruments \citep{bowden2016consistent}. The mode-based estimator uses
the empirical distribution of the ratio estimates and is consistent if the largest cluster of ratio estimates corresponds to valid instruments. Outlier-robust procedures instead focus on variants whose observed associations are poorly explained by the fitted IVW model. After fitting the IVW estimator, the residual for SNP $j$ can be written as, $r_j = \hat{\beta}_{Yj} - \hat{\theta}_{\mathrm{IVW}}\hat{\beta}_{Xj}.$ Instruments with large residuals have a larger impact on heterogeneity statistics such as Cochran's $Q$ statistic. 
The MR-Presso method \citep{verbanck2018detection} detects variants with unusually large residuals in the IVW regression and removes them before re-estimation. \cite{rees2019robust} proposed robust regression approaches to reduce the influence of instruments with heterogeneous effects.

An alternative strategy models the pleiotropic effects directly through penalized regression. Building on earlier work in the instrumental variables literature \citep{kang2016instrumental}, \citet{rees2019robust} proposed the MR-Lasso estimator for the summary-data setting \eqref{eq:summary_level_model}. In this approach, a separate intercept term is introduced for each genetic variant, allowing for instrument-specific pleiotropic effects. The causal effect and direct effects are estimated by solving

\begin{equation}
(\hat{\theta}_L, \hat{\alpha}_L) =
\arg\min_{\theta,\alpha}\sum_{j=1}^{J}\tilde{w}_j
(\hat{\beta}_{Yj} - \theta \hat{\beta}_{Xj} - \alpha_j)^2
+
\lambda\sum_{j=1}^{J} |\alpha_j|,
\label{eq:lasso}
\end{equation}

where $\lambda > 0$ is a regularization parameter selected through a heterogeneity stopping rule based on Cochran’s Q statistic. The $\ell_1$ penalty shrinks the direct-effect parameters toward zero and sets some of them exactly equal to zero. Variants with $\hat{\alpha}_j = 0$ are interpreted as valid instruments, whereas variants with nonzero estimates are treated as potentially invalid or outlying instruments \citep{rees2019robust}.

\section{Adaptive Penalized Mendelian Randomization}
\label{sec:method}

MR-Lasso provides a useful framework for identifying invalid instruments; however, MR-Lasso still inherits several limitations of the standard Lasso penalty. First, the penalty is applied uniformly across variants, regardless of the magnitude of their direct effects or even when the association estimates differ substantially in their levels of precision \cite{grant2021pleiotropy}. Such a penalty may not be optimal when the magnitude of pleiotropic effects varies across genetic variants. Second, the standard Lasso is known to introduce shrinkage bias in nonzero coefficients. Third, as shown by \cite{windmeijer2019use}, the Lasso estimator may fail to consistently identify invalid instruments when the so-called irrepresentable condition is violated. In MR studies, these limitations are especially relevant because the strengths of the instruments often vary substantially across variants. To address these issues, we propose an MR-ALasso estimator, which modifies the common penalty in \eqref{eq:lasso} by data-dependent weights that allow different variants to be penalized to different degrees. Specifically, we define
\begin{equation}
(\hat{\theta}_{\mathrm{AL}}, \hat{\boldsymbol{\alpha}}_{\mathrm{AL}})=
\arg\min_{\theta,\boldsymbol{\alpha}}
 \frac{1}{2}
\left(\hat{\boldsymbol{\beta}}_{Y} - \theta \hat{\boldsymbol{\beta}}_{X} - \boldsymbol{\alpha}\right)^{\top}\mathbf S
\left(\hat{\boldsymbol{\beta}}_{Y} - \theta \hat{\boldsymbol{\beta}}_{X} - \boldsymbol{\alpha}\right)
+ \lambda_n \sum_{j=1}^{J} \omega_j |\alpha_j|,
\label{eq:Alasso}
\end{equation}

where \(\omega_j>0\) are adaptive penalty weights, \(\hat{\boldsymbol{\beta}}_Y=(\hat{\beta}_{Y1},\ldots,\hat{\beta}_{YJ})^\top\), \(\hat{\boldsymbol{\beta}}_X=(\hat{\beta}_{X1},\ldots,\hat{\beta}_{XJ})^\top\), \(\boldsymbol{\alpha}=(\alpha_1,\ldots,\alpha_J)^\top\), and \(\mathbf S=\mathrm{diag}(\tilde w_1,\ldots,\tilde w_J)\). Equation \eqref{eq:Alasso} is not a standard adaptive Lasso \citep{zou2006adaptive} problem, because the causal-effect parameter $\theta$ is left unpenalized and only the direct effect parameters are penalized.

The weights are defined from an initial estimator of the causal effect. We use the median ratio estimator
$$\tilde{\theta}=\mathrm{median}
\left\{\frac{\hat{\beta}_{Yj}}{\hat{\beta}_{Xj}} :\hat{\beta}_{Xj} \neq 0
\right\},$$
which is consistent under standard regularity conditions when more than half
of the candidate instruments are valid. Using this initial estimator, we define preliminary direct-effect estimates of $\alpha_j$'s by
$$\tilde{\alpha}_j=\hat{\beta}_{Yj} - \tilde{\theta}\hat{\beta}_{Xj},
\qquad j = 1,\dots,J.$$
The adaptive weights are then defined as $\omega_j=\frac{1}{|\tilde{\alpha}_j|^\nu}$ for a given value of $\nu$ \citep{zou2006adaptive}. The optimization problem is a weighted Lasso problem in the direct-effect parameters $\alpha_j$ and can be solved using
standard Lasso algorithms. The interpretation of adaptive weights ($\omega_j$) is straightforward. If the initial direct-effect estimate $\tilde{\alpha}_j$ is small, then $\omega_j$ is large and the corresponding coefficient is penalized heavily, making it more likely that the final estimate is shrunk to zero. On the other hand, if $\tilde{\alpha}_j$ is large, then $\omega_j$ is smaller and the coefficient is penalized less strongly. In this way, the adaptive penalty aims to improve the separation between valid and invalid instruments.

To analyze the theoretical analysis of the MR-ALasso, we first profile out the causal parameter $\theta$ from the objective function in \eqref{eq:Alasso}. Theorem \ref{thm:reducedcriterion} represents that, after minimization w.r.t $\theta$, the MR-ALasso criterion can be written as a penalized quadratic function of $\boldsymbol\alpha$, with associated matrix $\hat{\mathbf C}_n$.

\begin{theorem}
Let $(\hat{\theta}_{\mathrm{AL}}, \hat{\boldsymbol{\alpha}}_{\mathrm{AL}})$ be defined by \eqref{eq:Alasso}. Then, after minimizing over $\theta$, $\hat{\boldsymbol{\alpha}}_{\mathrm{AL}}$ satisfies
\begin{equation}
\arg\min_{\boldsymbol\alpha}\frac{1}{2}
(\hat{\boldsymbol\beta}_Y-\boldsymbol\alpha)^\top
\hat{\mathbf C}_n
(\hat{\boldsymbol\beta}_Y-\boldsymbol\alpha)
+ \lambda_n \sum_{j=1}^{J} \omega_j |\alpha_j|,
\label{eq:Alasso-alpha}
\end{equation}
where $\hat{\mathbf C}_n=\mathbf S- \mathbf S\hat{\boldsymbol\beta}_X (\hat{\boldsymbol\beta}_X^\top \mathbf S\hat{\boldsymbol\beta}_X)^{-1} \hat{\boldsymbol\beta}_X^\top \mathbf S$ is the matrix obtained by profiling out the causal-effect parameter $\theta$ from the weighted least squares part of the objective.
\label{thm:reducedcriterion}
\end{theorem}

The proof is provided in Section \ref{proof:prop_reducedcrit} of the Supplementary Material.

\subsection{Choice of the tuning parameter}
\label{subsec:tuning}

The choice of the tuning parameter $\lambda_n$ affects the performance of MR-ALasso in estimating the causal effect. The regularization parameter $\lambda_n$ controls the level of sparsity. Larger values of $\lambda_n$ impose heavier shrinkage and hence more direct-effect estimates of the genetic variants are equal to zero. Since instruments with $\hat{\alpha}_j=0$ are treated as valid, the choice of $\lambda_n$ plays a crucial role in determining both the selected instrument set and the causal effect estimate. Although K-fold cross-validation can be used to select the tuning parameter $\lambda_n$, cross-validation is appropriate for prediction and may select models that include too many variables \citep{buhlmann2011statistics,windmeijer2019use}. In the present setting, this may lead to insufficient sparsity in the estimated direct-effect parameters and may therefore affect the
selection of valid instruments.

In the absence of an externally specified value of $\lambda_n$, we select the tuning parameter using a heterogeneity stopping rule \citep{rees2019robust}, motivated by Cochran's $Q$ statistic. To determine the value of $\lambda_n$, we fit the adaptive-lasso over a range of values along the solution path. For each value of $\lambda_n$, instruments with $\hat{\alpha}_j=0$ are identified as valid, and the IVW method is computed using only these selected instruments. We then compute the residual standard error (RSE) to assess the amount of residual heterogeneity among the selected instruments. Following \citet{rees2019robust}, we move along the solution path from weaker to stronger penalization and select the value of $\lambda_n$ immediately before the first sufficiently large increase in heterogeneity, defined by $\mathrm{RSE}>1$ and an increase in RSE between two consecutive values of $\lambda_n$ that exceeds a threshold derived from a chi-squared cutoff. If no such increase is found, we select the largest value of $\lambda_n$ along the computed solution path.

After selecting $\lambda_n$, we use a post-ALasso estimator to reduce the shrinkage bias induced by penalization \citep{belloni2012sparse}. Naive standard errors for the post-selection estimator may be obtained from the weighted regression fitted using only the valid SNPs. However, as in other post-selection procedures, these standard errors do not account for the uncertainty introduced by the selection step and may therefore underestimate the true variability. For this reason, inference based on the post-selection estimator should be interpreted carefully. To address this limitation and obtain more reliable statistical inference, we consider bootstrap smoothing for the computation of standard errors and confidence intervals in Subsection \ref{subsec:bootstrap_smoothing}.

\subsection{A selection consistency result for MR-ALasso}
\label{subsec:selection_consistency}
We now provide the selection consistency result for the proposed MR-ALasso estimator. The result is formulated for the two-sample summary-data model and shows that the adaptive penalty can recover the set of invalid instruments with probability tending to one.  We assume that the number of variants $J$ is fixed and that the exposure and outcome GWAS sample sizes increase at the same speed.

\begin{assumption}
\label{assump:pleiotropy_structure}
Suppose $\mathcal{A}=\{j:\alpha_j\neq 0\}$ and $\mathcal{V}=\mathcal A^c=\{j:\alpha_j=0\}$ denote the sets of invalid and valid genetic variants, respectively, and let $s=|\mathcal{A}|$.

(i) The number of invalid variants satisfies $s < J/2.$\\
(ii) If $\mathcal A\neq\emptyset$, the nonzero direct effects are bounded away
from zero, i.e.,
$$\min_{j\in\mathcal{A}} |\alpha_j| \ge c_\alpha$$
for some constant $c_\alpha>0$.\\
(iii) The SNP-exposure associations are bounded away from zero
$$\min_{1\le j\le J}|\beta_{Xj}|\ge c_X$$
for some constant $c_X>0$.
\end{assumption}

Assumption \ref{assump:pleiotropy_structure}(i) imposes a majority-valid condition, requiring fewer than half of the genetic variants to be invalid. This condition is related to the majority-valid assumptions used in robust MR and invalid instrument selection methods \citep{bowden2016consistent, kang2016instrumental}. Assumption \ref{assump:pleiotropy_structure}(ii) requires that the nonzero direct effects are sufficiently separated from zero, which is a standard minimum signal strength condition for consistent variable selection with penalized estimators \citep{zou2006adaptive}. Assumption \ref{assump:pleiotropy_structure}(iii) avoids weak or zero SNP-exposure associations in the construction of ratio estimates.

\begin{proposition}
Let $n=n_Y$ and suppose that $n_X/n_Y\to \kappa\in(0,\infty)$. Consider the two-sample summary-data model under Assumptions \ref{assump:two-sample_summary_data} and \ref{assump:pleiotropy_structure}. Assume that the initial estimator \(\tilde{\theta}\) is  consistent, i.e., $\sqrt n(\tilde{\theta}-\theta)=O_p(1)$. Let the adaptive weights be $$\omega_j=|\tilde{\alpha}_j|^{-\nu},
\qquad
\tilde{\alpha}_j=\hat{\beta}_{Yj}-\tilde{\theta}\hat{\beta}_{Xj},$$ with $\nu>0$. Assume further that $n^{-1}\hat{\mathbf C}_n\xrightarrow{p}\mathbf C,$ where the principal submatrix $\mathbf C_{\mathcal A\mathcal A}$ is positive definite. If the tuning parameter satisfies 
\begin{equation}
\lambda_n=o(\sqrt n),
\qquad
n^{(\nu-1)/2}\lambda_n\to\infty,
\label{eq:tuning_prop}
\end{equation}
then the MR-ALasso estimator consistently identifies the set of invalid variants: $$P(\widehat{\mathcal A}_n=\mathcal A)\to 1,$$ where $\widehat{\mathcal{A}}_n=\{j:\hat{\alpha}_{\mathrm{AL},j}\neq 0\}.$
\label{prop:MR-ALasso}
\end{proposition}
The proof is provided in Section \ref{proof:selectionconsist} of the Supplementary Material.

The proposition \ref{prop:MR-ALasso} shows that the adaptive penalty can exploit the asymptotic separation between valid and invalid variants. As a consequence, the post-MR-ALasso estimator is asymptotically equivalent to the oracle IVW estimator based on the true valid instruments; see Corollary \ref{cor:oracle_equiv}.

\begin{corollary}
Under the assumptions of Proposition \ref{prop:MR-ALasso}, let $\widehat{\mathcal V}_n=\widehat{\mathcal A}_n^c$ be the selected valid set. Then the post-MR-ALasso estimator $\hat{\theta}_{\mathrm{post\mbox{-}AL}}$, obtained by applying IVW to the selected valid set $\widehat{\mathcal V}_n$, is asymptotically identical to the oracle IVW estimator $\hat{\theta}_{\mathrm{oracle}}$, based on the true valid set $\mathcal V$, with probability tending to one.
\label{cor:oracle_equiv}
\end{corollary}
The proof is provided in Section \ref{proof:oracle_equiv} of the Supplementary Material.

\subsection{Bootstrap smoothing}
\label{subsec:bootstrap_smoothing}

Previous work has shown that the MR-Lasso procedure may suffer from inflated type-I error and unstable selection of invalid instruments \citep{rees2019robust,grant2021pleiotropy}. Although the proposed MR-ALasso estimator improves upon MR-Lasso by using adaptive penalization, the resulting post-selection estimator remains a discontinuous function of the summary statistics. A small perturbation of $(\hat{\boldsymbol\beta}_X,\hat{\boldsymbol\beta}_Y)$ may change the set of instruments selected as valid, and hence may change the post-selection IVW estimator. Consequently, standard errors computed after conditioning on a single selected set may fail to account for the uncertainty introduced by the data-adaptive selection step. This issue is particularly important under the null hypothesis or when some direct effects are close to the valid--invalid selection boundary. We investigate the finite-sample impact of this post-selection uncertainty in the simulation study in Section~\ref{sec:simulation}, where empirical type-I error and power are reported under several pleiotropic scenarios.

To address this problem, we propose a bootstrap-smoothed extension of MR-ALasso, denoted by MR-ALasso-B. Bootstrap smoothing is a form of model averaging \citep{efron2014estimation} and is closely related to bagging in the prediction literature \citep{buja2006observations}. It is also related to recent resampling-based approaches for improving post-selection inference in summary-data MR \citep{xie2026winner}. The goal is not to construct a nonparametric bootstrap for the individual-level exposure and outcome GWAS samples, which are unavailable in a summary-data MR analysis. Instead, the proposed procedure uses SNP-level resampling to perturb the empirical instrument set, repeatedly applies the MR-ALasso selection rule, and averages the resulting post-selection IVW estimators. Thus, MR-ALasso-B is proposed as a finite-sample smoothing procedure for reducing the instability of post-selection inference.

Let $$\mathcal W_J=\left\{\mathbf w=(w_1,\ldots,w_J)^\top\in\mathbb N_0^J: \sum_{j=1}^J w_j=J \right\}$$ denote the set of all bootstrap count vectors. For $\mathbf w\in\mathcal W_J$, define $$\mathcal I(\mathbf w)=\{j:w_j>0\}$$ as the set of instruments used in the bootstrap sample. For a generic bootstrap count vector $\mathbf w\in\mathcal W_J$, let $\widehat{\mathcal V}(\mathbf w)$ denote the set of instruments selected as valid by MR-ALasso when applied to the bootstrap instrument set determined by $\mathbf w$, mapped back to the original SNP indices. For bootstrap replicate $b=1,\ldots,B_0$, we generate  $\mathbf W_b^* = (W_{1b}^*,\ldots,W_{Jb}^*)^\top \sim \mathrm{Multinomial} \left( J;\frac{1}{J},\ldots,\frac{1}{J} \right),$ where $W_{jb}^*$ represents the number of times SNP $j$ appears in bootstrap replicate $b$. This multinomial count representation is equivalent to resampling the $J$ instruments with replacement and provides a convenient formulation for bootstrap smoothing and related delta-method calculations \citep{efron2014estimation}. For the $b$th bootstrap replicate, we apply MR-ALasso to obtain the selected set of valid instruments, denoted by $\widehat{\mathcal V}(\mathbf W_b^*)$. We then compute the post-selection IVW estimate
\begin{equation}
   \hat\theta(\mathbf W_b^*)=\frac{\sum_{j\in\widehat{\mathcal V}(\mathbf W_b^*)} W_{jb}^* \hat\beta_{Xj}\hat\beta_{Yj}\hat\sigma_{Yj}^{-2}} {\sum_{j\in\widehat{\mathcal V}(\mathbf W_b^*)} W_{jb}^* \hat\beta_{Xj}^2 \hat\sigma_{Yj}^{-2}}.
   \label{eq:bootstrap_post_ivw}
\end{equation}
Bootstrap estimates for which fewer than two instruments are selected as valid are excluded. The retained bootstrap estimates are relabelled as $b=1,\ldots,B$. Thus, $B$ denotes the number of retained bootstrap replicates. The MR-ALasso-B estimator is then defined by
\begin{equation}
    \tilde\theta_{\mathrm{AL\mbox{-}B}}=\frac{1}{B}\sum_{b=1}^B \hat\theta_b,
\label{eq:post-IVW-b}
\end{equation}
where $\hat\theta_b= \hat\theta(\mathbf W_b^*)$ for bootstrap replicates.
The estimator in \eqref{eq:post-IVW-b} is a smoothed version of the post-MR-ALasso estimator. Instead of conditioning on a single selected valid set, MR-ALasso-B averages over many bootstrap-selected valid sets. This model-averaging step is intended to reduce the instability of the post-selection estimator and to partially incorporate the additional uncertainty introduced by data-adaptive instrument selection. Bootstrap-based stabilization ideas for lasso-type procedures have also been studied previously, for example, in the Bolasso of \citet{bach2008bolasso}.  
\subsubsection{Statistical properties}

Let $P_*$ denote probability with respect to the multinomial bootstrap distribution, conditional on the observed summary statistics. For $\mathbf w\in\mathcal W_J$, define
$$
D(\mathbf w)=\sum_{j\in\widehat{\mathcal V}(\mathbf w)} w_j\hat\beta_{Xj}^2\hat\sigma_{Yj}^{-2}
$$ 
and
$$\mathcal R=\left\{ \mathbf w\in\mathcal W_J: |\widehat{\mathcal V}(\mathbf w)|\ge2,\quad D(\mathbf w)>0 \right\}.
$$ 
Thus, $\mathcal R$ is the set of bootstrap count vectors for
which the post-selection IVW in \eqref{eq:bootstrap_post_ivw} is
defined. 

The variance estimator used in MR-ALasso-B is an adaptation of the nonparametric delta-method formula of \citet{efron2014estimation} for bootstrap-smoothed estimators. In the present setting, the bootstrap counts are the SNP counts $W_j^*$, and the bootstrap replication is the post-selection IVW estimate $\hat\theta(\mathbf W^*)$. Because the post-selection IVW ratio is computed only for bootstrap draws in $\mathcal R$, the corresponding covariance is taken conditionally on $\mathcal R$.

In the following result, $\mathcal R$ is treated as a fixed subset of $\mathcal W_J$, determined by the observed summary statistics and by the MR-ALasso selection rule.

\begin{theorem}
\label{thm:mralassob_efron_delta}
Conditional on the observed summary statistics, let
$$
\mathbf W^*(p)=(W_1^*(p),\ldots,W_J^*(p))^\top\sim\mathrm{Multinomial}(J;p_1,\ldots,p_J),
$$
where $p=(p_1,\ldots,p_J)$ is a probability vector, and let $p_0=(1/J,\ldots,1/J)$. Suppose that $P_{p_0}(\mathcal R)>0$, and define the conditional bootstrap-smoothed functional
$$
S_{\mathcal R}(p)=E_p\left\{\hat \theta (\mathbf W^*(p))\mid\mathbf W^*(p)\in\mathcal R\right\}.
$$ 
Then the nonparametric delta-method variance of
$S_{\mathcal R}(p)$ at $p=p_0$ is
\begin{equation}
    \sum_{j=1}^J\left(S_j^*\right)^2,
    \label{eq:efron_var_mralassob}
\end{equation}

where $S_j^*=\mathrm{Cov}_{p_0}\left[W_j^*,
\hat \theta(\mathbf W^*)\mid\mathbf W^*\in\mathcal R\right],$
or equivalently,
$S_j^*=\mathrm{Cov}_*\left[W_j^*,\hat\theta(\mathbf W^*)\mid\mathcal R\right].$
\end{theorem}
The proof is provided in the Supplementary Material, Section \ref{supp:mralassob_efron_delta}.

The estimated variance of $\tilde\theta_{\mathrm{AL\mbox{-}B}}$ is the analogue of \eqref{eq:efron_var_mralassob},
\begin{equation}
    \widehat{\mathrm{Var}}\left(\tilde\theta_{\mathrm{AL\mbox{-}B}} \right)= \sum_{j=1}^J \widehat S_j^2,
    \label{eq:est_var_mralassob}    
\end{equation}
where $\widehat S_j=\frac{1}{B}\sum_{b=1}^B\left(W_{jb}^*-\bar W_j^*\right)\left\{\hat\theta(\mathbf W_b^*) - \tilde\theta_{\mathrm{AL\mbox{-}B}} \right\}$ and $\bar W_j^*=\frac{1}{B}\sum_{b=1}^BW_{jb}^*$. A Wald-type confidence interval is then constructed as
$$
\tilde\theta_{\mathrm{AL\mbox{-}B}}\pm z_{1-\alpha/2}\left\{\widehat{\mathrm{Var}}\left(\tilde\theta_{\mathrm{AL\mbox{-}B}}\right) \right\}^{1/2}.$$
The variance estimator measures the sensitivity of the smoothed post-selection estimator to perturbations of the empirical instrument distribution. It is therefore used to capture variation arising from both the post-selection IVW estimates and the instability of the selected valid sets across bootstrap replicates. It should not be interpreted as a nonparametric bootstrap variance for the individual-level GWAS samples. The implementation of the MR-ALasso-B framework is provided in Algorithm \ref{alg:mr_alasso_b}.
\begin{algorithm}[t]
\small
\DontPrintSemicolon
\SetAlgoLined
\caption{MR-ALasso-B}
\label{alg:mr_alasso_b}
\KwIn{
Summary data
$\{(\hat\beta_{Xj},\hat\sigma_{Xj},
\hat\beta_{Yj},\hat\sigma_{Yj})\}_{j=1}^J$, bootstrap replicates, threshold $\tau$
}
\KwOut{
$\tilde\theta_{\mathrm{AL\mbox{-}B}}$, $\widehat\sigma_{\mathrm{AL\mbox{-}B}}$, confidence interval, aggregated valid set $\widehat{\mathcal V}_{\mathrm{agg}}$}

\For{$b=1,\dots,B$}{
Draw
$$
\mathbf W_b^*=(W_{1b}^*,\dots,W_{Jb}^*)^\top\sim\mathrm{Multinomial}(J;1/J,\dots,1/J)
$$\;
Construct the bootstrap instrument set by resampling SNPs according to $\mathbf W_b^*$\;
Apply MR-ALasso to the bootstrap instrument set\;
Map the selected valid instruments back to the original SNP indices and obtain $\widehat{\mathcal V}(\mathbf W_b^*)$\;
\If{$|\widehat{\mathcal V}(\mathbf W_b^*)|\ge2$}{
Compute the weighted post-selection IVW estimate
$\hat\theta(\mathbf W_b^*)$ using \eqref{eq:bootstrap_post_ivw}\;
}
\Else{
Set $\hat\theta(\mathbf W_b^*)$ to missing\;
}
}
Compute the causal estimator $\tilde\theta_{\mathrm{AL\mbox{-}B}}=B^{-1}\sum_{b=1}^B \hat\theta(\mathbf W_b^*)$\;
Compute selection frequencies $\hat\pi_j=B^{-1}\sum_{b=1}^B I\{j\in\widehat{\mathcal V}(\mathbf W_b^*)\} $ and define $\widehat{\mathcal V}_{\mathrm{agg}}=\{j:\hat\pi_j\ge\tau\}$\;
Compute $\widehat S_j=B^{-1}\sum_{b=1}^B(W_{jb}^*-\bar W_j^*)\left\{\hat\theta(\mathbf W_b^*)-\tilde\theta_{\mathrm{AL\mbox{-}B}}\right\}$\;
Estimate $\widehat\sigma_{\mathrm{AL\mbox{-}B}}^2=\sum_{j=1}^J\widehat S_j^2$\;
Construct the confidence interval
$$\tilde\theta_{\mathrm{AL\mbox{-}B}} \pm z_{1-\alpha/2} \sqrt{\widehat\sigma_{\mathrm{AL\mbox{-}B}}^2} $$
\end{algorithm}

\begin{proposition}
\label{prop:efron_variance_limit}
Conditional on the observed summary statistics, suppose that $P_*(\mathcal R)>0$ and that 
$$E_*\left[\hat\theta(\mathbf W^*)^2\mid\mathcal R\right]<\infty.$$ 
Then, as $B\to\infty$, $ \widehat S_j \xrightarrow{P_*} S_j^*$, where $ S_j^* = \mathrm{Cov}_* \left[ W_j^*, \hat\theta(\mathbf W^*) \mid \mathcal R \right].$ Consequently, since $J$ is fixed,
$$
\sum_{j=1}^J \widehat S_j^2 \xrightarrow{P_*} \sum_{j=1}^J (S_j^*)^2.
$$
\end{proposition}

The proof is provided in the Supplementary Material, Section~\ref{supp:variance_limit}.

Proposition \ref{prop:efron_variance_limit} shows that the empirical variance estimator in \eqref{eq:est_var_mralassob}  converges to the SNP-level delta-method variance target in \eqref{eq:efron_var_mralassob}. This is the component of uncertainty targeted by MR-ALasso-B and is directly related to the instability of post-selection inference observed in finite samples.

MR-ALasso-B also summarizes the stability of instrument selection. For each SNP $j$, define the bootstrap selection frequency
$$
\hat\pi_j = \frac{1}{B} \sum_{b=1}^B I\{j\in\widehat{\mathcal V}(\mathbf W_b^*)\}.
$$
The aggregated valid set is $\widehat{\mathcal V}_{\mathrm{agg}} = \{j:\hat\pi_j\ge\tau\},$ where $\tau\in(0,1)$ is a pre-specified threshold. In our implementation, we use $\tau=0.5$, corresponding to a majority-vote rule. SNPs not contained in $\widehat{\mathcal V}_{\mathrm{agg}}$ are reported as invalid. This aggregation rule is also related to stability selection \citep{meinshausen2010stability,shah2013variable}, but the formal error-control guarantees from stability selection do not apply directly here. The present procedure uses SNP-level bootstrap perturbations and summarizes the selection of valid instruments rather than active regression variables. Therefore, $\widehat{\mathcal V}_{\mathrm{agg}}$ is used as a descriptive measure of classification stability. As in the stability-selection literature, the threshold can be viewed as a tuning parameter controlling the trade-off between classification stability and the risk of misclassifying instruments.

\section{Simulation Study}
\label{sec:simulation}
We considered two simulation designs to evaluate the performance of the proposed methods. Model-I is based on an individual-level structural model and is used to assess estimation and inference under pleiotropic settings, including both balanced and directional pleiotropy, and both InSIDE-satisfying and InSIDE-violating scenarios. Model-II is a summary-data design motivated by Proposition~\ref{prop:MR-ALasso} and Corollary~\ref{cor:oracle_equiv}, and is used to examine whether the finite-sample behavior of the estimator is consistent with the corresponding theoretical arguments, particularly with respect to invalid-instrument identification.

\subsection{Model-I}
The Model-I simulation study is based on an individual-level structural model and assesses the performance of the competing MR methods across different types of pleiotropic cases. The  data-generating process is similar to those used by \citet{slob2020comparison} and \citet{grant2021pleiotropy}. For each individual $i=1,\dots,n$, we generated $J=100$ independent genetic variants $Z_{ij}$, an unobserved confounder $U_i$, an exposure $X_i$, and an outcome $Y_i$ according to
\begin{align*}
U_i &= \sum_{j=1}^J \phi_j Z_{ij} + \varepsilon_{Ui}, \\
X_i &= \sum_{j=1}^J \gamma_j Z_{ij} + U_i + \varepsilon_{Xi}, \\
Y_i &= \sum_{j=1}^J \alpha_j Z_{ij} + \theta X_i + U_i + \varepsilon_{Yi}.
\end{align*}
The error terms $\varepsilon_{Ui}$, $\varepsilon_{Xi}$, and $\varepsilon_{Yi}$ were generated independently from standard normal distributions. The genetic variants were simulated as $Z_{ij} \sim \mathrm{Binomial}(2,\mathrm{maf}_j),$ $\mathrm{maf}_j \sim \mathrm{Uniform}(0.1,0.5),$ independently across variants. Two independent samples were generated an exposure sample of size $n_X$ and an outcome sample of size $n_Y$, with $n_Y=n_X=20000$. In the exposure sample, SNP-exposure associations were obtained by fitting linear regressions of $X$ on each variant separately. In the outcome sample, SNP-outcome associations were obtained by fitting linear regressions of $Y$ on each variant separately. This yielded the two-sample summary data $(\hat{\beta}_{Xj},\hat{\beta}_{Yj})$ and their standard errors. We considered $p_I\in\{0.15,0.30,0.45\},$ where $p_I$ denotes the proportion of invalid instruments, to investigate settings with increasing proportions of invalid instruments. We also considered two instrument strength designs. In the equal-strength design, valid and invalid instruments had the same average SNP-exposure association, $\gamma_{\mathcal A}=\gamma_{\mathcal V}$. In the invalid-stronger design, invalid instruments had stronger SNP-exposure associations, $\gamma_{\mathcal A}=3\gamma_{\mathcal V}$, following \citet{windmeijer2019use}. Further, the vector of exposure effects $\gamma=(\gamma_1,\dots,\gamma_J)$ was then adjusted so that the variants explained a fixed 12$\%$ of the variance in the exposure.

We considered the following four pleiotropic scenarios:
\begin{enumerate}
\item {Balanced pleiotropy, InSIDE satisfied}: The pleiotropic effects for invalid variants were generated from $\alpha_j \sim N(0,0.2^2)$. The parameter $\phi_j$ was set to zero for all variants so that the InSIDE assumption was satisfied; 
\item {Directional pleiotropy, InSIDE satisfied}: Directional pleiotropy was introduced by generating $\alpha_j \sim N(0.1,0.2^2)$ for the invalid variants. The parameter $\phi_j$ was set to zero for all variants so that the InSIDE assumption is satisfied; 

\item {Directional pleiotropy, InSIDE violated}: The pleiotropic effects follow $\alpha_j \sim N(0.1,0.2^2)$, while $\phi_j$ is generated from $\phi_j \sim U(0,0.1)$. The nonzero values of $\phi_j$ induce
a violation of the InSIDE assumption;
\item {Balanced pleiotropy, InSIDE violated}: The $\alpha_j$'s corresponding to the invalid instruments were generated from $\alpha_j \sim N(0,0.2^2)$ so that they remain centered at zero, whereas $\phi_j \sim U(-0.05,0.05)$ introduces dependence that violates the InSIDE assumption.
\end{enumerate}

In scenarios 1 and 2, we assess the performance of the methods when pleiotropy is present, but the InSIDE assumption remains valid. In scenarios 3 and 4, the InSIDE assumption is violated because of the direct association of the genetic variants with the confounder \citep{grant2021pleiotropy}. These scenarios are useful for studying how the methods perform when pleiotropic effects operate partly through pathways involving the unobserved confounder. Inference was conducted at the 5\% significance level. We considered two values for the true causal effect, $\theta \in \{0,\,0.2\}$, corresponding to the null case and a moderate causal effect, respectively. The simulation scenario was replicated 2000 times.


For comparison, we considered the following MR methods: the standard inverse-variance weighted estimator (MR-IVW), a robust regression method based on MM-estimation with Tukey's bisquare objective function (MR-Robust), MR-Egger regression (MR-Egger), the weighted median estimator (MR-Median), the weighted mode estimator (MR-Mode), and the summary-data Lasso procedure (MR-Lasso). These methods were implemented using the R package \texttt{MendelianRandomization} \citep{yavorska2017mendelianrandomization,broadbent2020mendelianrandomization}. The proposed MR-ALasso and MR-ALasso-B methods were implemented in the \texttt{MRAlasso} package \citep{qasim_mralasso_2026}. The performance of each method was assessed using several criteria, including bias, the standard deviation (SD) of the causal effect estimates, root mean squared error (RMSE), empirical coverage probability (CP), and empirical power. For $\theta=0$, we report empirical type I error. For $\theta=0.20$, we report empirical coverage probability and empirical power. We used RMSE as the primary summary measure to compare methods, since it combines information on both bias and variability. However, the relative importance of each performance measure depends on the aim of the analysis.

Figure \ref{fig:model-I-boxplot} shows the distribution of the estimated causal effects for the MR methods across the four pleiotropy scenarios when the invalid instruments are stronger than the valid instruments for $p_I=0.30$. The boxplots summarize the empirical distribution of the estimates across simulation replicates. Outlying points were omitted from the display to improve readability, although all simulated estimates were retained when computing the performance measures. Figure \ref{fig:model-I-boxplot} shows that MR-ALasso-B and MR-ALasso had the smallest RMSE in all four scenarios. The boxplots also indicate that MR-ALasso-B estimates had the less dispersion than those from the other methods, notably under directional pleiotropy and violation of the InSIDE assumption (Scenario 3).

\begin{figure}[htbp]
    \centering 
    \begin{subfigure}{\textwidth}
        \centering
        \includegraphics[width=\textwidth,height=0.44\textheight,keepaspectratio]{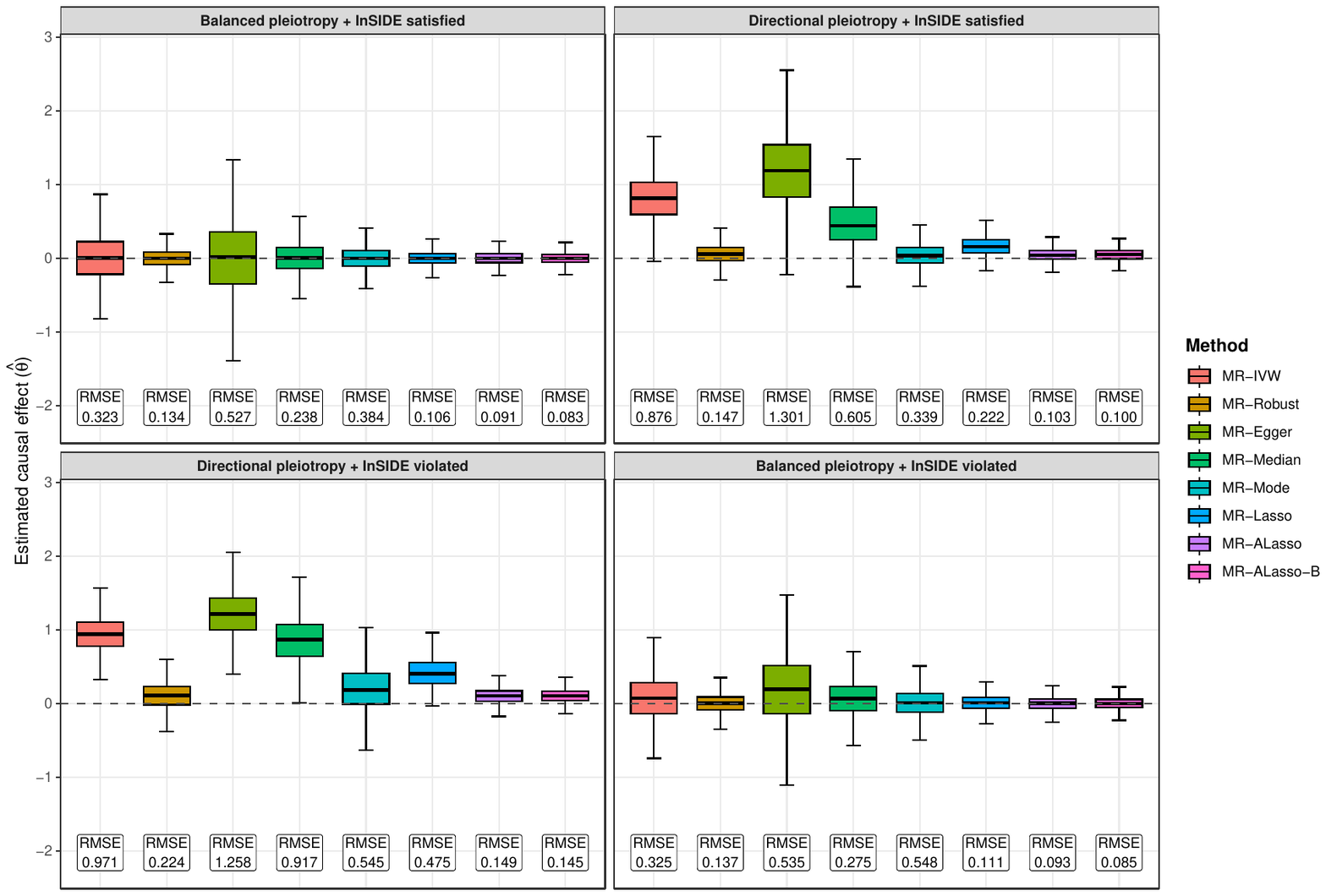}
        \caption{True causal effect $\theta = 0$.}
        \label{fig:theta0}
    \end{subfigure}
    
    \begin{subfigure}{\textwidth}
        \centering
        \includegraphics[width=\textwidth,height=0.45\textheight,keepaspectratio]{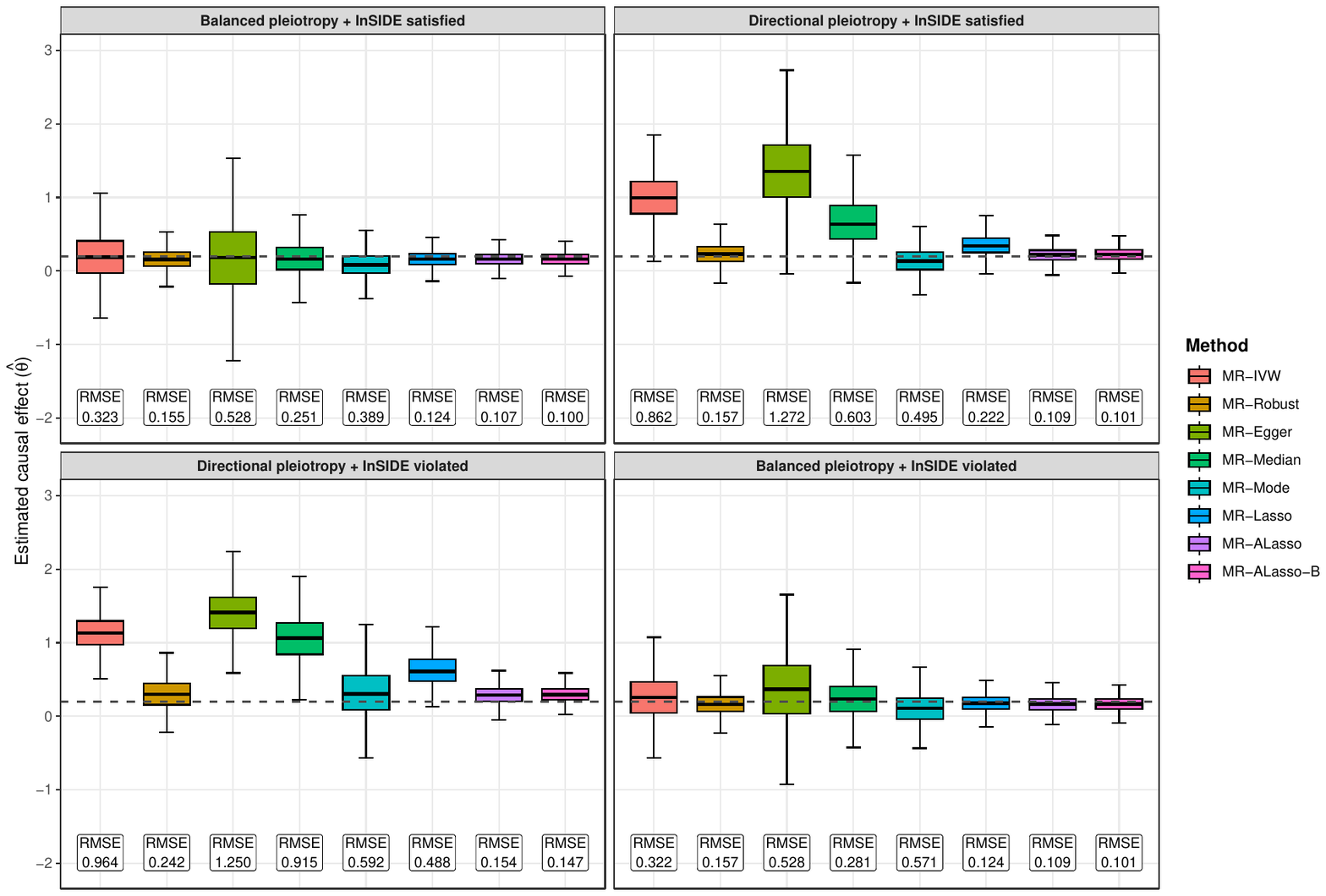}
        \caption{True causal effect $\theta = 0.2$.}
        \label{fig:theta02}
    \end{subfigure}
    
    \caption{Distribution of estimated causal effects under different pleiotropy scenarios in Model-I. The dashed horizontal line indicates the true causal effect. RMSE values are shown below the corresponding boxplots.}
    \label{fig:model-I-boxplot}
\end{figure}

Tables \ref{tab:modelI_theta0_stronger_invalid} and \ref{tab:model_I-theta0.2} report the simulation results for the invalid-stronger design under the null hypothesis $\theta=0$ and the equal-strength design under the alternative hypothesis $\theta=0.20$, respectively. The remaining results are given in the Supplementary Material. The simulation results show that the proposed MR-ALasso improved upon MR-Lasso in estimation accuracy, whereas MR-ALasso-B improved statistical inference after instrument selection. These gains were most noticeable in the more difficult settings with directional pleiotropy, violation of the InSIDE assumption, and stronger invalid instruments.

Table \ref{tab:modelI_theta0_stronger_invalid} shows that the standard MR-IVW estimator performed poorly in Scenarios 2 and 3. In these settings, MR-IVW had a large bias and severe inflation of type-I error. MR-Egger also performed poorly, with large RMSE and unstable estimation. MR-Median and MR-Mode methods were more stable than MR-IVW in some settings, but their behavior was not uniform across all scenarios. MR-Robust performed well in the milder settings, but its performance weakened as the invalid-instrument structure became more difficult. Scenario 3 was the most challenging null setting, since pleiotropic effects were directional and the InSIDE assumption was violated. In this case, MR-IVW, MR-Egger, and MR-Median all performed poorly, with very large bias and type-I error close to one in several cases. MR-Robust was less biased than those methods, but its type-I error still became too large when the proportion of invalid instruments increased.

In addition, these results also showed a clear distinction between the proposed methods. MR-ALasso generally had smaller SD and RMSE than MR-Lasso, indicating that the adaptive penalty improved the estimation step. This was especially clear when invalid instruments were stronger than valid instruments, where ordinary Lasso penalization was more affected by the presence of strong invalid variants. However, both MR-Lasso and MR-ALasso remained anti-conservative under the null, with empirical coverage probabilities below the nominal level and type-I error rates that increased as the proportion of invalid instruments increased. On the other hand, MR-ALasso-B had much better type-I error control and coverage than the lasso-based procedures. This pattern was consistent across all four pleiotropy scenarios. In the invalid-stronger design, MR-ALasso-B also achieved the smallest RMSE among the proposed methods in Table \ref{tab:modelI_theta0_stronger_invalid}, which indicates that the gain in inference was not obtained at the cost of poor estimation. Similar findings were observed in the equal-strength design (Table \ref{tab:model-I-theta0}), where, for example, at 45\% invalid instruments, MR-ALasso-B had a type-I error of 6\%, which was close to the nominal significance level. MR-Robust gave the second smallest type-I error at 10\%, but its RMSE was significantly higher than that of the proposed methods.

Table \ref{tab:model_I-theta0.2} reports the simulation results when the true causal effect was positive. In Scenario 1, MR-Robust, MR-Lasso, MR-ALasso, and MR-ALasso-B all produced small bias and small RMSE, whereas MR-Egger showed large bias and MR-IVW yielded higher RMSE than the robust and Lasso-based methods. MR-Median also performed reasonably well, while MR-Mode had lower power and larger RMSE. Under directional pleiotropy with InSIDE satisfied (Scenario 2), MR-IVW became increasingly biased as the proportion of invalid instruments increased, and the MR-Median estimator, despite high power, showed increasing bias at 30\% and 45\% invalid instruments. In this setting, MR-ALasso improved upon MR-Lasso at all invalid-instrument proportions through smaller bias and smaller RMSE, which is consistent with the adaptive penalty providing better downweighting of variants with larger direct effects. MR-ALasso-B had similar RMSE to MR-ALasso and gave higher empirical coverage probabilities, although with some reduction in power at the largest proportion. In Scenario 3, where pleiotropic effects were directional and the InSIDE assumption was violated. In this case, MR-IVW, MR-Egger, and MR-Median showed large bias and very large RMSE, whereas MR-Robust, although much better than these methods, also showed poor performance as the proportion of invalid instruments increased. MR-Lasso retained full or nearly full power, but exhibited large bias and poor coverage. While MR-ALasso substantially reduced both bias and RMSE relative to MR-Lasso. MR-ALasso-B had slightly larger bias than MR-ALasso in some cases, but it provided higher coverage probabilities and lower RMSE than standard methods. Finally, under balanced pleiotropy with violation of InSIDE (Scenario 4), Lasso-based procedures performed well relative to the other MR methods. In addition, MR-ALasso improved slightly upon MR-Lasso, and MR-ALasso-B provided the highest coverage probabilities. Overall, Model-I shows that MR-ALasso improves estimation relative to MR-Lasso, while MR-ALasso-B improves post-selection inference by increasing coverage and reducing anti-conservative behavior. 

\begin{table}[H]
\caption{Simulation results for Model-I under different pleiotropy scenarios when $\theta=0$ and $\gamma_{\mathcal{A}} = 3\gamma_{\mathcal{V}}$.}
\centering
\resizebox{\textwidth}{!}{
\begin{tabular}{lcccc cccc cccc}
\toprule
\multirow{3}{*}{Method}
& \multicolumn{4}{c}{15\% Invalid}
& \multicolumn{4}{c}{30\% Invalid}
& \multicolumn{4}{c}{45\% Invalid} \\
\cmidrule(lr){2-5} \cmidrule(lr){6-9} \cmidrule(lr){10-13}
& $\hat{\theta}$ & SD & RMSE & Type-I
& $\hat{\theta}$ & SD & RMSE & Type-I
& $\hat{\theta}$ & SD & RMSE & Type-I \\
\midrule

\multicolumn{13}{l}{\textit{Scenario 1: Balanced pleiotropy, InSIDE satisfied}} \\

MR-IVW      & 0.000 & 0.281 & 0.281 & 0.35 & 0.004 & 0.323 & 0.323 & 0.25 & -0.003 & 0.339 & 0.339 & 0.17 \\
MR-Robust   & 0.000 & 0.062 & 0.062 & 0.06 & 0.001 & 0.134 & 0.134 & 0.06 & 0.008 & 0.289 & 0.289 & 0.08 \\
MR-Egger    & -0.002 & 0.540 & 0.540 & 0.39 & 0.004 & 0.527 & 0.527 & 0.21 & -0.008 & 0.522 & 0.522 & 0.09 \\
MR-Median   & 0.003 & 0.158 & 0.158 & 0.36 & 0.002 & 0.238 & 0.238 & 0.44 & 0.001 & 0.300 & 0.300 & 0.49 \\
MR-Mode     & -0.002 & 0.161 & 0.161 & 0.06 & 0.002 & 0.384 & 0.384 & 0.05 & -0.018 & 0.833 & 0.833 & 0.08 \\
MR-Lasso    & 0.000 & 0.063 & 0.063 & 0.19 & 0.001 & 0.106 & 0.106 & 0.32 & -0.002 & 0.164 & 0.164 & 0.44 \\
MR-ALasso   & 0.000 & 0.058 & 0.058 & 0.16 & 0.000 & 0.091 & 0.091 & 0.26 & -0.001 & 0.134 & 0.134 & 0.36 \\
MR-ALasso-B & 0.001 & 0.055 & 0.055 & 0.03 & 0.000 & 0.083 & 0.083 & 0.03 & -0.002 & 0.126 & 0.126 & 0.05 \\

\midrule
\multicolumn{13}{l}{\textit{Scenario 2: Directional pleiotropy, InSIDE satisfied}} \\

MR-IVW      & 0.506 & 0.282 & 0.579 & 0.81 & 0.814 & 0.325 & 0.876 & 0.91 & 1.042 & 0.341 & 1.096 & 0.95 \\
MR-Robust   & 0.009 & 0.063 & 0.063 & 0.06 & 0.059 & 0.135 & 0.147 & 0.07 & 0.320 & 0.326 & 0.457 & 0.17 \\
MR-Egger    & 0.922 & 0.543 & 1.069 & 0.80 & 1.187 & 0.531 & 1.301 & 0.82 & 1.312 & 0.527 & 1.414 & 0.80 \\
MR-Median   & 0.214 & 0.198 & 0.292 & 0.65 & 0.497 & 0.345 & 0.605 & 0.85 & 0.763 & 0.407 & 0.865 & 0.94 \\
MR-Mode     & 0.015 & 0.161 & 0.162 & 0.06 & 0.043 & 0.336 & 0.339 & 0.07 & 0.071 & 0.665 & 0.669 & 0.09 \\
MR-Lasso    & 0.038 & 0.067 & 0.077 & 0.26 & 0.172 & 0.141 & 0.222 & 0.66 & 0.411 & 0.257 & 0.485 & 0.90 \\
MR-ALasso   & 0.011 & 0.060 & 0.061 & 0.18 & 0.045 & 0.093 & 0.103 & 0.31 & 0.137 & 0.146 & 0.200 & 0.55 \\
MR-ALasso-B & 0.013 & 0.055 & 0.057 & 0.03 & 0.048 & 0.088 & 0.100 & 0.04 & 0.148 & 0.137 & 0.202 & 0.09 \\

\midrule
\multicolumn{13}{l}{\textit{Scenario 3: Directional pleiotropy, InSIDE violated}} \\

MR-IVW      & 0.691 & 0.248 & 0.734 & 0.97 & 0.941 & 0.242 & 0.971 & 1.00 & 1.085 & 0.226 & 1.108 & 1.00 \\
MR-Robust   & 0.020 & 0.073 & 0.076 & 0.11 & 0.118 & 0.190 & 0.224 & 0.21 & 0.542 & 0.385 & 0.665 & 0.55 \\
MR-Egger    & 1.046 & 0.382 & 1.114 & 0.97 & 1.215 & 0.327 & 1.258 & 0.99 & 1.294 & 0.294 & 1.327 & 1.00 \\
MR-Median   & 0.543 & 0.321 & 0.631 & 0.94 & 0.859 & 0.319 & 0.917 & 0.99 & 1.025 & 0.296 & 1.067 & 1.00 \\
MR-Mode     & 0.110 & 0.712 & 0.720 & 0.23 & 0.223 & 0.498 & 0.545 & 0.27 & 0.288 & 1.503 & 1.530 & 0.29 \\
MR-Lasso    & 0.110 & 0.095 & 0.146 & 0.61 & 0.426 & 0.210 & 0.475 & 0.96 & 0.746 & 0.255 & 0.789 & 1.00 \\
MR-ALasso   & 0.025 & 0.066 & 0.071 & 0.29 & 0.104 & 0.107 & 0.149 & 0.58 & 0.277 & 0.164 & 0.322 & 0.87 \\
MR-ALasso-B & 0.025 & 0.060 & 0.065 & 0.04 & 0.114 & 0.099 & 0.151 & 0.10 & 0.283 & 0.141 & 0.317 & 0.26 \\

\midrule
\multicolumn{13}{l}{\textit{Scenario 4: Balanced pleiotropy, InSIDE violated}} \\

MR-IVW      & 0.036 & 0.280 & 0.282 & 0.36 & 0.069 & 0.317 & 0.325 & 0.26 & 0.096 & 0.323 & 0.337 & 0.19 \\
MR-Robust   & 0.003 & 0.064 & 0.064 & 0.06 & 0.003 & 0.137 & 0.137 & 0.08 & 0.036 & 0.293 & 0.295 & 0.12 \\
MR-Egger    & 0.107 & 0.523 & 0.533 & 0.40 & 0.185 & 0.502 & 0.535 & 0.25 & 0.263 & 0.486 & 0.552 & 0.18 \\
MR-Median   & 0.034 & 0.178 & 0.181 & 0.41 & 0.073 & 0.265 & 0.275 & 0.51 & 0.117 & 0.316 & 0.337 & 0.55 \\
MR-Mode     & 0.058 & 1.904 & 1.904 & 0.09 & 0.009 & 0.548 & 0.548 & 0.10 & 0.038 & 0.919 & 0.919 & 0.08 \\
MR-Lasso    & 0.004 & 0.064 & 0.065 & 0.21 & 0.014 & 0.110 & 0.111 & 0.35 & 0.037 & 0.169 & 0.173 & 0.47 \\
MR-ALasso   & 0.003 & 0.060 & 0.060 & 0.19 & 0.002 & 0.093 & 0.093 & 0.29 & 0.006 & 0.134 & 0.134 & 0.38 \\
MR-ALasso-B & 0.000 & 0.055 & 0.055 & 0.04 & 0.002 & 0.083 & 0.083 & 0.04 & 0.005 & 0.122 & 0.122 & 0.03 \\

\bottomrule
\end{tabular}
}
\begin{tablenotes}[flushleft]
\scriptsize
\item Note: $\hat{\theta}$ is the empirical mean of the estimated causal effect; SD is the standard deviation; RMSE is the root mean squared error; and Type-I error is the empirical rejection probability under the null hypothesis $\theta=0$; and $\gamma_{\mathcal{A}} = 3\gamma_{\mathcal{V}}$ is the case when invalid instruments are stronger than valid instruments.
\end{tablenotes}
\label{tab:modelI_theta0_stronger_invalid}
\end{table}

\begin{table}[H]
\centering
\caption{Simulation results for Model-I under different pleiotropy scenarios when $\theta=0.20$ and $\gamma_{\mathcal{A}} = \gamma_{\mathcal{V}}$.}
\resizebox{\textwidth}{!}{
\begin{tabular}{lccccc ccccc ccccc}
\toprule
\multirow{2}{*}{Method}
& \multicolumn{5}{c}{15\% Invalid}
& \multicolumn{5}{c}{30\% Invalid}
& \multicolumn{5}{c}{45\% Invalid} \\
\cmidrule(lr){2-6} \cmidrule(lr){7-11} \cmidrule(lr){12-16}
& $\hat{\theta}$ & SD & RMSE & CP & Power
& $\hat{\theta}$ & SD & RMSE & CP & Power
& $\hat{\theta}$ & SD & RMSE & CP & Power \\
\midrule

\multicolumn{16}{l}{\textit{Scenario 1: Balanced pleiotropy, InSIDE satisfied}} \\
MR-IVW       & 0.184 & 0.146 & 0.147 & 0.94 & 0.30 & 0.187 & 0.206 & 0.206 & 0.93 & 0.18 & 0.183 & 0.249 & 0.249 & 0.94 & 0.13 \\
MR-Robust    & 0.184 & 0.038 & 0.041 & 0.92 & 1.00 & 0.184 & 0.052 & 0.054 & 0.94 & 0.94 & 0.187 & 0.092 & 0.093 & 0.97 & 0.53 \\
MR-Egger     & 0.048 & 0.429 & 0.455 & 0.92 & 0.06 & 0.065 & 0.609 & 0.624 & 0.94 & 0.06 & 0.054 & 0.752 & 0.765 & 0.93 & 0.06 \\
MR-Median    & 0.175 & 0.050 & 0.056 & 0.94 & 0.94 & 0.176 & 0.061 & 0.065 & 0.92 & 0.85 & 0.177 & 0.078 & 0.082 & 0.89 & 0.73 \\
MR-Mode      & 0.151 & 0.095 & 0.106 & 0.97 & 0.34 & 0.148 & 0.261 & 0.266 & 0.97 & 0.35 & 0.148 & 0.242 & 0.247 & 0.96 & 0.34 \\
MR-Lasso     & 0.185 & 0.040 & 0.043 & 0.89 & 1.00 & 0.184 & 0.048 & 0.050 & 0.88 & 0.99 & 0.184 & 0.062 & 0.064 & 0.84 & 0.94 \\
MR-ALasso    & 0.185 & 0.039 & 0.042 & 0.90 & 1.00 & 0.185 & 0.047 & 0.049 & 0.89 & 0.99 & 0.186 & 0.060 & 0.062 & 0.86 & 0.94 \\
MR-ALasso-B  & 0.186 & 0.039 & 0.041 & 0.97 & 0.99 & 0.186 & 0.046 & 0.048 & 0.97 & 0.93 & 0.187 & 0.058 & 0.060 & 0.97 & 0.80 \\

\midrule
\multicolumn{16}{l}{\textit{Scenario 2: Directional pleiotropy, InSIDE satisfied}} \\
MR-IVW       & 0.434 & 0.149 & 0.278 & 0.68 & 0.86 & 0.684 & 0.210 & 0.528 & 0.35 & 0.92 & 0.929 & 0.254 & 0.772 & 0.16 & 0.97 \\
MR-Robust    & 0.187 & 0.039 & 0.041 & 0.94 & 1.00 & 0.201 & 0.052 & 0.052 & 0.95 & 0.96 & 0.280 & 0.101 & 0.129 & 0.97 & 0.72 \\
MR-Egger     & 0.059 & 0.475 & 0.495 & 0.93 & 0.07 & 0.075 & 0.659 & 0.671 & 0.94 & 0.06 & 0.083 & 0.808 & 0.816 & 0.94 & 0.06 \\
MR-Median    & 0.199 & 0.051 & 0.051 & 0.96 & 0.98 & 0.236 & 0.063 & 0.073 & 0.91 & 0.97 & 0.296 & 0.088 & 0.131 & 0.74 & 0.97 \\
MR-Mode      & 0.152 & 0.102 & 0.113 & 0.97 & 0.34 & 0.152 & 0.318 & 0.322 & 0.97 & 0.36 & 0.150 & 0.247 & 0.252 & 0.96 & 0.36 \\
MR-Lasso     & 0.191 & 0.041 & 0.041 & 0.91 & 1.00 & 0.207 & 0.051 & 0.051 & 0.89 & 1.00 & 0.250 & 0.072 & 0.088 & 0.73 & 0.99 \\
MR-ALasso    & 0.188 & 0.040 & 0.041 & 0.91 & 1.00 & 0.194 & 0.048 & 0.049 & 0.90 & 0.99 & 0.216 & 0.065 & 0.067 & 0.84 & 0.97 \\
MR-ALasso-B  & 0.190 & 0.040 & 0.041 & 0.97 & 0.99 & 0.198 & 0.047 & 0.047 & 0.97 & 0.95 & 0.222 & 0.064 & 0.068 & 0.98 & 0.87 \\

\midrule
\multicolumn{16}{l}{\textit{Scenario 3: Directional pleiotropy, InSIDE violated}} \\
MR-IVW       & 0.731 & 0.202 & 0.568 & 0.07 & 0.99 & 1.033 & 0.221 & 0.862 & 0.01 & 1.00 & 1.220 & 0.216 & 1.043 & 0.00 & 1.00 \\
MR-Robust    & 0.194 & 0.043 & 0.044 & 0.94 & 0.99 & 0.230 & 0.075 & 0.080 & 0.96 & 0.87 & 0.453 & 0.188 & 0.315 & 0.88 & 0.52 \\
MR-Egger     & 1.486 & 0.566 & 1.405 & 0.11 & 0.94 & 1.764 & 0.518 & 1.647 & 0.05 & 0.97 & 1.798 & 0.461 & 1.663 & 0.04 & 0.99 \\
MR-Median    & 0.315 & 0.089 & 0.146 & 0.52 & 1.00 & 0.574 & 0.226 & 0.437 & 0.10 & 1.00 & 0.902 & 0.303 & 0.765 & 0.01 & 1.00 \\
MR-Mode      & 0.150 & 0.218 & 0.224 & 0.96 & 0.33 & 0.160 & 0.095 & 0.103 & 0.93 & 0.36 & 0.184 & 0.229 & 0.229 & 0.89 & 0.45 \\
MR-Lasso     & 0.210 & 0.046 & 0.047 & 0.86 & 1.00 & 0.294 & 0.081 & 0.124 & 0.44 & 1.00 & 0.502 & 0.161 & 0.342 & 0.06 & 1.00 \\
MR-ALasso    & 0.194 & 0.043 & 0.043 & 0.89 & 1.00 & 0.214 & 0.058 & 0.060 & 0.81 & 0.99 & 0.275 & 0.085 & 0.113 & 0.56 & 0.99 \\
MR-ALasso-B  & 0.196 & 0.042 & 0.042 & 0.97 & 0.99 & 0.219 & 0.056 & 0.059 & 0.97 & 0.94 & 0.287 & 0.081 & 0.119 & 0.94 & 0.81 \\

\midrule
\multicolumn{16}{l}{\textit{Scenario 4: Balanced pleiotropy, InSIDE violated}} \\
MR-IVW       & 0.224 & 0.157 & 0.158 & 0.92 & 0.40 & 0.257 & 0.214 & 0.221 & 0.90 & 0.31 & 0.285 & 0.246 & 0.260 & 0.91 & 0.26 \\
MR-Robust    & 0.186 & 0.039 & 0.041 & 0.92 & 1.00 & 0.187 & 0.054 & 0.056 & 0.94 & 0.92 & 0.201 & 0.099 & 0.099 & 0.97 & 0.53 \\
MR-Egger     & 0.291 & 0.572 & 0.579 & 0.79 & 0.26 & 0.426 & 0.653 & 0.690 & 0.82 & 0.24 & 0.533 & 0.677 & 0.755 & 0.82 & 0.25 \\
MR-Median    & 0.184 & 0.052 & 0.055 & 0.94 & 0.95 & 0.193 & 0.071 & 0.072 & 0.90 & 0.87 & 0.207 & 0.095 & 0.095 & 0.84 & 0.80 \\
MR-Mode      & 0.156 & 0.412 & 0.415 & 0.99 & 0.20 & 0.175 & 2.384 & 2.384 & 0.99 & 0.10 & 0.120 & 1.139 & 1.141 & 0.99 & 0.05 \\
MR-Lasso     & 0.187 & 0.040 & 0.042 & 0.90 & 1.00 & 0.188 & 0.050 & 0.051 & 0.86 & 0.99 & 0.194 & 0.063 & 0.063 & 0.84 & 0.95 \\
MR-ALasso    & 0.187 & 0.040 & 0.042 & 0.90 & 1.00 & 0.186 & 0.049 & 0.051 & 0.88 & 0.99 & 0.188 & 0.061 & 0.062 & 0.86 & 0.95 \\
MR-ALasso-B  & 0.187 & 0.039 & 0.041 & 0.96 & 0.99 & 0.187 & 0.048 & 0.050 & 0.96 & 0.93 & 0.189 & 0.058 & 0.059 & 0.97 & 0.78 \\
\bottomrule
\end{tabular}
}
\begin{tablenotes}[flushleft]
\scriptsize
\item Note: $\hat{\theta}$ denotes the empirical mean of the estimated causal effect; SD is the empirical standard deviation; RMSE is the root mean squared error; CP is the empirical coverage probability of the nominal 95\% confidence interval; and Power is the empirical probability of rejecting the null hypothesis $H_0:\theta=0$ when the true causal effect is $\theta=0.20$. The condition $\gamma_{\mathcal{A}} = \gamma_{\mathcal{V}}$ indicates that valid and invalid instruments have equal strength on average.
\end{tablenotes}

\label{tab:model_I-theta0.2}
\end{table}

\subsection{Model-II}

Model-II was considered as a summary-data simulation study to examine the finite-sample implications of Proposition \ref{prop:MR-ALasso} and Corollary \ref{cor:oracle_equiv}. Specifically, it was used to assess whether adaptive penalization improves the identification of invalid instruments relative to MR-Lasso, and whether the resulting post-selection estimator behaves similarly to the oracle IVW benchmark when the majority-valid condition holds. Unlike Model-I, which was created primarily to assess overall estimation and inference under a more realistic individual-level data-generating process, Model-II was considered to study invalid instrument identification and oracle-type post-selection behavior under directly simulated summary-data associations. Following \citet{zhao2020statistical}, we generated summary statistics directly under the independence condition corresponding to their Assumption 1. Summary associations were generated as $\hat{\beta}_{Xj} \sim N(\beta_{Xj},\mathrm{SE}(\hat{\beta}_{Xj})^2)$ and $\hat{\beta}_{Yj} \sim N(\beta_{Yj},\mathrm{SE}(\hat{\beta}_{Yj})^2)$, where the SNP--exposure effects were generated as $\beta_{Xj}=\gamma_j\sim |N(\mu_\gamma,\sigma_\gamma^2)|$. We set $\mu_\gamma=0.1$ and $\sigma_\gamma=0.02$. The SNP--outcome associations followed the structural model $\beta_{Yj}=\theta\beta_{Xj}+\alpha_j$, where $\theta \in\{0, 0.2\}$ denotes the true causal effect and $\alpha_j$ represents the direct (pleiotropic) effect of variant $j$ on the outcome.

For each genetic variant $j=1,\dots,100$, the minor allele frequency (maf) was generated independently from a uniform distribution $\text{maf}_j \sim \mathrm{Uniform}(0.1,0.5)$ \citep{slob2020comparison}. The variance of variant $j$ was therefore $\mathrm{Var}(Z_j)=2\,\text{maf}_j(1-\text{maf}_j)$. The standard errors of the summary association estimates were set to $\mathrm{SE}(\hat{\beta}_{Xj})=(\sigma_X^2/(n_X\,\mathrm{Var}(Z_j)))^{1/2}$ and $\mathrm{SE}(\hat{\beta}_{Yj})=(\sigma_Y^2/(n_Y\,\mathrm{Var}(Z_j)))^{1/2}$, where we set $\sigma_X^2=\sigma_Y^2=1$ and considered two sample sizes, $n_X=n_Y\in\{10000, 20000\}$. A proportion $p_I$ of the variants was assumed to be invalid instruments, for which the direct effect was set to $\alpha_j = 0.2$, while the remaining variants were valid instruments. We considered $p_I\in\{0.15,0.30,0.45\}$ to investigate scenarios with increasing levels of invalid instruments. We also considered two designs for the instrument strengths; (i) in the \textit{equal strength design}, all variants $\gamma_j$ were generated from the same distribution, (ii) in the \textit{invalid stronger design}, the invalid variants were made stronger than valid instruments, $\gamma_{\mathcal{A}} = 3\gamma_{\mathcal{V}}$.

This simulation design was chosen to represent the assumptions underlying Proposition \ref{prop:MR-ALasso}. First, the proportion of invalid instruments was set to $p_I\in\{0.15,0.30,0.45\}$, so that fewer than half of the candidate instruments were invalid and the majority-valid condition held. This design is also consistent with the sparse invalid instrument setting commonly considered in the instrumental variable literature. In the individual-level instrumental variable model, \citet{kang2016instrumental} showed that identification and estimation are possible when fewer than half of the candidate instruments are invalid, and \citet{windmeijer2019use} considered an ALasso procedure under the same majority-valid condition. Second, the direct effects of invalid variants were fixed at a nonzero value, which ensured the identifiability of valid and invalid instruments, corresponding to the minimum signal condition in the proposition.

Table \ref{tab:modelII_main_30} reports the Model-II results for the 30\% invalid-instrument setting, which corresponds most closely to the motivating design in \citet{windmeijer2019use}. The full set of Model-II results, including the 15\% and 45\% invalid instrument settings, is provided in the supplementary materials, see Tables \ref{tab:modelII_theta0.0} and \ref{tab:modelII_theta0.2}. In addition to the usual estimation summaries, this simulation design also reports the average number of variants identified as invalid, $\bar{s}_{\mathrm{sel}}$, and the proportion of replicates in which all invalid variants were correctly identified, ``All invalid''. These two measures were included because Model-II was designed to examine the selection behavior of the Lasso-based procedures more directly than Model-I. It can be seen that in the equal strength design, MR-ALasso had a bias close to zero, and low RMSE under both $\theta=0$ and $\theta=0.20$. Moreover, $\bar{s}_{\mathrm{sel}}$ was close to the true number of invalid instruments and the frequency of all invalid instruments selected was equal to one when $p_I=0.30$. On the other hand, MR-Lasso was less stable. Even when its bias was small, it often excluded more instruments than were truly invalid, indicating that valid variants were often incorrectly removed. This contrast became more significant in the invalid stronger design, where MR-Lasso selected too many variants as invalid and, in some settings, exhibited very large bias and RMSE, whereas MR-ALasso showed much more stability. These findings also support Corollary \ref{cor:oracle_equiv}. In several settings, the post-MR-ALasso estimator was numerically very close to the Oracle IVW benchmark in both bias and RMSE, especially in the equal strength design and in the 30\% invalid-instrument setting shown in Table \ref{tab:modelII_main_30}. This pattern is consistent with the theoretical conclusion that, when the invalid set is correctly identified, post-selection IVW should behave similarly to the oracle estimator based on the true valid set. The naive MR-IVW estimator, which used all instruments without a selection step, showed substantial bias when invalid variants were included. The simulation results for MR-ALasso-B were more mixed in this design. Because Model-II was constructed primarily to examine invalid-instrument identification and oracle-type post-selection behavior, the inferential performance of MR-ALasso-B is assessed mainly under Model-I.

\begin{table}[t]
\caption{Simulation results for Model-II when $p_I=0.30$.}
\centering
\resizebox{\textwidth}{!}{
\begin{tabular}{lll ccccc ccccc}
\toprule
\multirow{2}{*}{$\theta$} & \multirow{2}{*}{Design} & \multirow{2}{*}{Method}
& \multicolumn{5}{c}{$n=10000$}
& \multicolumn{5}{c}{$n=20000$} \\
\cmidrule(lr){4-8} \cmidrule(lr){9-13}
& & & $\hat{\theta}$ & SD & RMSE & $\bar{s}_{\mathrm{sel}}$ & All invalid
& $\hat{\theta}$ & SD & RMSE & $\bar{s}_{\mathrm{sel}}$ & All invalid \\
\midrule

\multirow{10}{*}{$0$}
& \multirow{5}{*}{$\gamma_{\mathcal A}=\gamma_{\mathcal V}$}
& MR-IVW       & 0.563 & 0.037 & 0.564 &   --   &  --    & 0.570 & 0.034 & 0.571 &  --    &  --    \\
& & Oracle IVW  & 0.000 & 0.018 & 0.018 & 30 & 1 & 0.000 & 0.013 & 0.013 & 30 & 1 \\
& & MR-Lasso    & 0.005 & 0.020 & 0.021 & 30.9 & 1 & 0.003 & 0.014 & 0.015 & 30.8 & 1 \\
& & MR-ALasso   & 0.002 & 0.019 & 0.019 & 30.6 & 1 & 0.001 & 0.013 & 0.014 & 30.1 & 1 \\
& & MR-ALasso-B & 0.005 & 0.019 & 0.019 & 30.7 & 1 & 0.003 & 0.013 & 0.014 & 30.7 & 1 \\
\cmidrule(lr){4-13}
& \multirow{5}{*}{$\gamma_{\mathcal A}=3\gamma_{\mathcal V}$}
& MR-IVW       & 0.506 & 0.016 & 0.507 &   --   &  --    & 0.508 & 0.015 & 0.508 &   --   &  --    \\
& & Oracle IVW  & 0.000 & 0.018 & 0.018 & 30 & 1 & 0.000 & 0.013 & 0.013 & 30 & 1 \\
& & MR-Lasso    & 0.509 & 0.035 & 0.510 & 79.6 & 0 & 0.520 & 0.033 & 0.521 & 86.6 & 0 \\
& & MR-ALasso   & 0.003 & 0.019 & 0.020 & 30.7 & 1 & 0.002 & 0.013 & 0.014 & 30.1 & 1 \\
& & MR-ALasso-B & 0.008 & 0.020 & 0.022 & 31.0 & 1 & 0.004 & 0.014 & 0.014 & 30.7 & 1 \\

\midrule

\multirow{10}{*}{$0.20$}
& \multirow{5}{*}{$\gamma_{\mathcal A}=\gamma_{\mathcal V}$}
& MR-IVW       & 0.758 & 0.038 & 0.559 &   --   &   --   & 0.767 & 0.035 & 0.568 &   --   &   --   \\
& & Oracle IVW  & 0.195 & 0.018 & 0.019 & 30 & 1 & 0.198 & 0.013 & 0.013 & 30 & 1 \\
& & MR-Lasso    & 0.202 & 0.022 & 0.022 & 31.3 & 1 & 0.202 & 0.016 & 0.016 & 31.1 & 1 \\
& & MR-ALasso   & 0.198 & 0.020 & 0.020 & 30.8 & 1 & 0.199 & 0.014 & 0.014 & 30.9 & 1 \\
& & MR-ALasso-B & 0.202 & 0.019 & 0.019 & 31.0 & 1 & 0.202 & 0.014 & 0.014 & 31.0 & 1 \\
\cmidrule(lr){4-13}
& \multirow{5}{*}{$\gamma_{\mathcal A}=3\gamma_{\mathcal V}$}
& MR-IVW       & 0.705 & 0.017 & 0.505 &  --    &  --    & 0.707 & 0.015 & 0.508 &  --    &   --   \\
& & Oracle IVW  & 0.195 & 0.018 & 0.019 & 30 & 1 & 0.198 & 0.013 & 0.013 & 30 & 1 \\
& & MR-Lasso    & 0.706 & 0.035 & 0.507 & 78.5 & 0.001 & 0.718 & 0.034 & 0.519 & 86.2 & 0 \\
& & MR-ALasso   & 0.200 & 0.021 & 0.021 & 30.5 & 1 & 0.200 & 0.014 & 0.014 & 30.1 & 1 \\
& & MR-ALasso-B & 0.207 & 0.021 & 0.022 & 31.3 & 1 & 0.203 & 0.014 & 0.014 & 31.0 & 1 \\

\bottomrule
\end{tabular}
}
\begin{tablenotes}[flushleft]
\scriptsize
\item Note: Oracle IVW is the MR-IVW estimator computed using the true valid set of instruments. $\hat{\theta}$ denotes the empirical mean of the estimated causal effect; SD is the empirical standard deviation; RMSE is the root mean squared error; $\bar{s}_{\mathrm{sel}}$ is the average number of instruments identified as invalid; and ``All invalid'' is the empirical proportion of replicates in which the selected invalid set contains all truly invalid instruments, i.e., $\mathcal A\subseteq\widehat{\mathcal A}_n$.
\end{tablenotes}
\label{tab:modelII_main_30}
\end{table}

The benefit of MR-ALasso is also evident in the invalid-stronger design $(\gamma_{\mathcal A}=3\gamma_{\mathcal V})$. In this setting, MR-Lasso often performed poorly, especially when $p_I$ was moderate or large. By comparison, MR-ALasso showed a more stable performance. When $p_I=0.15$ and $p_I=0.30$, MR-ALasso had bias close to zero, smaller RMSE than MR-Lasso, and nearly perfect recovery of the invalid set. When $p_I=0.45$, where performance became more difficult for all methods, MR-ALasso remained much more accurate than MR-Lasso, especially at the larger sample size. These results are consistent with the role of the adaptive penalty, which penalizes variants with small apparent direct effects more strongly while allowing variants with clear evidence of nonzero direct effects to be selected as invalid. This improves the separation between valid and invalid instruments and is in line with Proposition \ref{prop:MR-ALasso}.

\section{Real-Data Application}
\label{sec:application}

To evaluate the practical performance of the proposed MR methods, we conducted large-scale bidirectional pairwise MR analyses across 22 complex traits \citep{hu2022mendelian}. The selected traits span a broad spectrum of anthropometric, cardiometabolic, psychiatric, behavioral, and social phenotypes that have been extensively studied in complex-trait human genetics. spanning cardiovascular, anthropometric, immune, neurological/psychiatric, and social domains \citep{hu2022mendelian}. For each trait, candidate IVs were selected by applying a GWAS p-value threshold $(5 \times 10^{-8})$ and subsequently pruned for linkage disequilibrium using an $r^2$ threshold of 0.001 to ensure instrument independence. Bidirectional MR analyses were conducted for all pairs using MR-IVW, MR-Robust, MR-Egger, MR-Median, MR-Mode, MR-Lasso, MR-ALasso, and MR-ALasso-B. 

\begin{figure}[H]
    \centering
    \includegraphics[width=\textwidth]{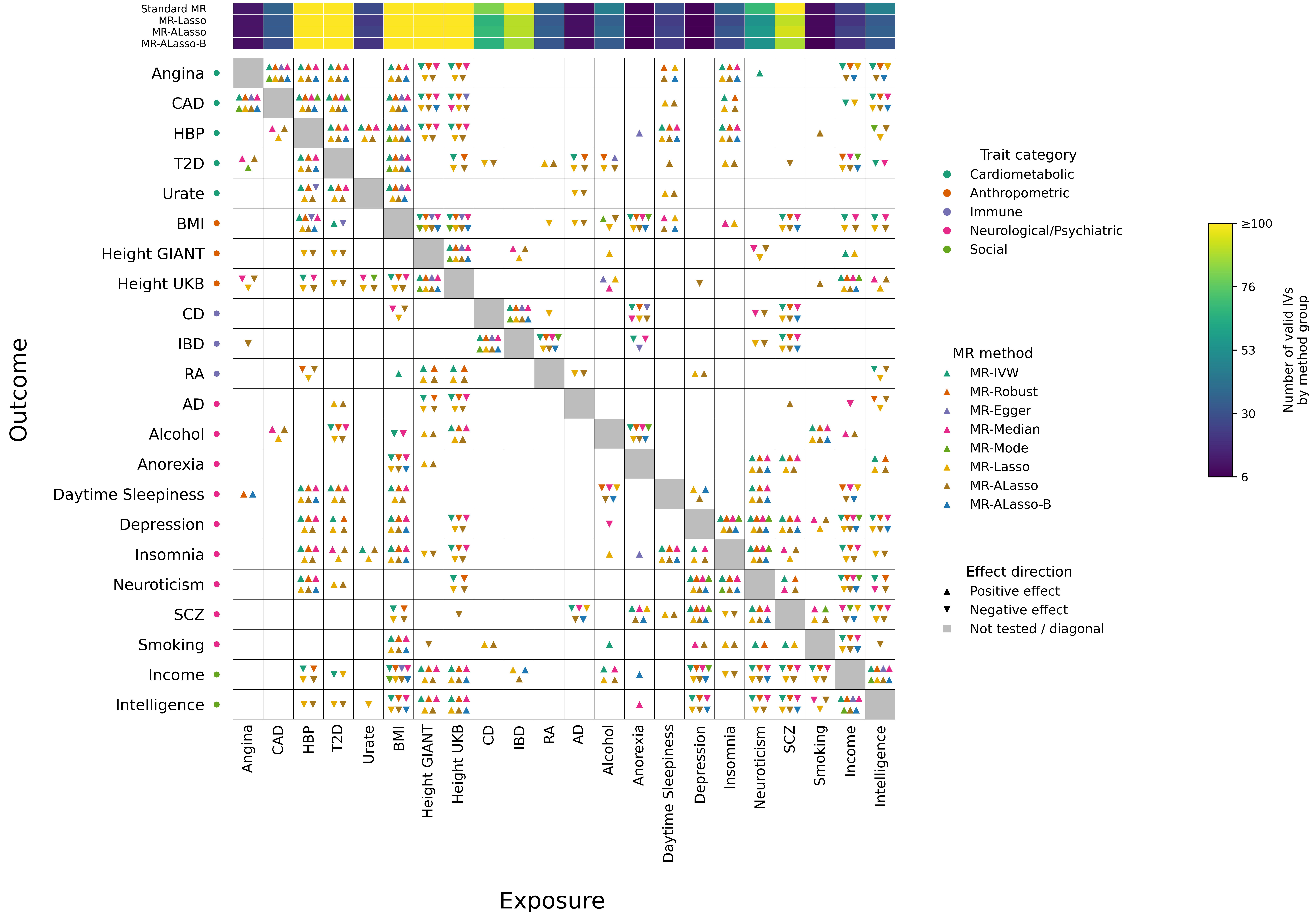}
    \caption{Bidirectional MR matrix summarizing results for 22 complex traits using multiple MR methods. The matrix presents pairwise MR results between exposure (rows) and outcome (columns) traits. Traits are categorized into cardiometabolic, anthropometric, immune, neurological/psychiatric, and social groups. Each cell shows colored triangles for MR methods with an FDR-adjusted $p$-value less than 0.05. Upward and downward triangles indicate positive and negative causal effects, respectively. The horizontal color bars indicate the number of instrumental variables used in the analysis. Trait abbreviations: CAD, coronary artery disease; HBP, high blood pressure; T2D, type 2 diabetes; BMI, body mass index; CD, Crohn’s disease; IBD, inflammatory bowel disease; RA, rheumatoid arthritis; AD, Alzheimer’s disease; SCZ, schizophrenia; Height GIANT, height from the GIANT consortium GWAS; Height UKB, height from the UK Biobank GWAS.}
    \label{fig:Bidirectional_MR_matrix.}
\end{figure}

\begin{figure}[H]
    \centering
    \includegraphics[width=\textwidth]{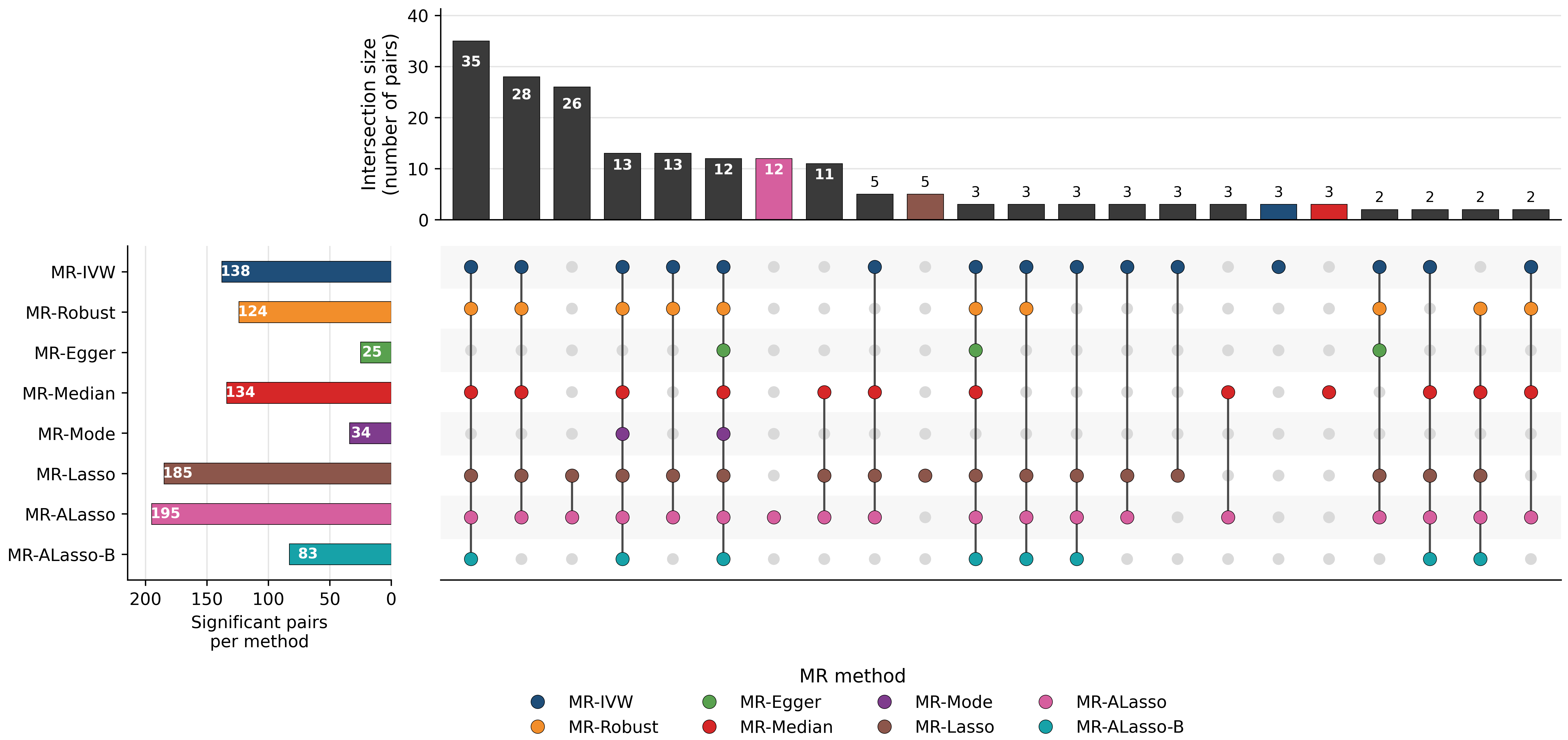}
    \caption{UpSet plot summarizing overlap of significant associations $( p < 0.05)$ across MR methods. Horizontal bars show the number of significant exposure--outcome pairs identified by each method, whereas vertical bars show the size of each intersection. Connected dots indicate the specific MR methods contributing to each overlap pattern. MR-Lasso and MR-ALasso identified more significant pairs, whereas MR-ALasso-B yielded a smaller set of associations.}
    \label{fig:UpSet_plot}
\end{figure}

Figures \ref{fig:Bidirectional_MR_matrix.} and \ref{fig:UpSet_plot} summarize the bidirectional pairwise MR results across all methods. Figure \ref{fig:Bidirectional_MR_matrix.} presents an exposure-outcome matrix, with each cell indicating the direction of association across MR methods for each exposure-outcome pair. Figure \ref{fig:UpSet_plot} summarizes the number of significant associations detected by each method and the agreement among them. Overall, MR-ALasso and MR-Lasso identified the largest number of significant associations, but MR-ALasso-B yielded a noticeably smaller discovery set. Importantly, a large proportion of MR-ALasso-B findings were replicated by multiple MR methods. For example, 35 of the 83 significant pairs detected by MR-ALasso-B ($42\%$) were also detected by IVW, MR-Robust, MR-Median, MR-LASSO, and MR-ALasso. Furthermore, twelve pairs were significant across all methods, including MR-Egger and MR-Mode, indicating that a core subset of associations identified by MR-ALasso-B is strongly supported by multiple MR frameworks. These results indicate that MR-ALasso-B filters out many findings unique to individual estimators while retaining those most reproducible across analytic frameworks. MR-ALasso and MR-Lasso identified the largest numbers of significant pairs, including several associations that were not detected by the standard MR methods. This pattern may reflect a combination of improved sensitivity and the anti-conservative behavior observed for these methods in the simulation study.

The smaller overlap of MR-ALasso-B findings with MR-Egger and MR-Mode is consistent with the different identifying assumptions and empirical behavior of these estimators. MR-Egger can provide valid causal estimates in the presence of directional pleiotropy only under the InSIDE assumption \citep{Bowden2015}. In addition, the interpretation of directional pleiotropy in MR-Egger depends on the allele-coding convention, and MR-Egger estimates can be sensitive to this coding choice \citep{dudbridge2025getting}. The MR-Mode estimator is robust under the zero modal pleiotropy assumption \citep{hartwig2017robust}, but mode-based estimation may have reduced sensitivity when the largest cluster of ratio estimates is weakly separated from competing clusters or when the causal signal is diffuse. This interpretation is consistent with the simulation results, where MR-Egger and MR-Mode often produced lower power than the proposed methods. 

We observed a core set of associations supported by multiple MR methods, alongside method-specific findings. The agreement between conventional MR methods and Lasso-based methods, together with the recovery of well-established relationships among cardiometabolic and anthropometric traits, suggests that several of the detected associations are biologically plausible. MR-Lasso and MR-ALasso identified a larger number of significant associations than the conventional MR methods, and many of these associations were directionally consistent with estimates from standard methods. This pattern suggests that Lasso-based selection of valid instruments may improve sensitivity for detecting potential causal relationships. At the same time, MR-Lasso and MR-ALasso also identified several associations that were not consistently supported by the conventional MR methods, which is consistent with the anti-conservative behavior observed for these methods in the simulation studies.

Several associations retained by MR-ALasso-B are consistent with evidence from randomized trials and large-scale clinical studies. MR-ALasso-B preserved the positive effect of BMI on T2D risk, as demonstrated by an intensive lifestyle intervention targeting modest weight loss, which reduced the incidence of T2D by 58\% \citep{diabetes2002reduction}. Similarly, MR-ALasso-B successfully identified the effect of BMI on blood pressure, consistent with the meta-analysis of 35 randomized weight-loss trials, which showed that a mean reduction of \(2.27\,\mathrm{kg}/\mathrm{m}^2\) in BMI lowered systolic and diastolic blood pressure by 5.79 mm Hg and 3.36 mm Hg, respectively \citep{yang2023effect}. Additionally, MR-ALasso-B retained causal effects of lipid metabolism on coronary artery disease, consistent with the large body of literature supporting the causal effect of elevated LDL cholesterol on cardiovascular disease \citep{mortensen2023low}. Reducing LDL cholesterol with simvastatin therapy can significantly lower all-cause mortality and major coronary events compared to placebo \citep{scandinavian1994randomised}. Consistently, MR-ALasso-B retained the harmful effects of hypertension on cardiovascular traits, consistent with the results of the antihypertensive trials, where blood pressure-lowering therapies reduce the risk of stroke by about 30\% and all-cause mortality by 13\% \citep{ettehad2016blood}. These observations suggest that the associations retained by MR-ALasso-B are consistent with established epidemiological and randomized-trial evidence.

MR-ALasso is expected to improve instrument selection over MR Lasso because the adaptive penalty allows differential shrinkage across variants, thereby improving the separation between valid and invalid instruments. MR-ALasso-B further stabilizes post-selection inference by averaging over bootstrap perturbations of the instrument set and by using a variance estimator that reflects sensitivity to the selected valid set. Taken together, the real-data analysis is consistent with the simulation results, MR-ALasso-B yields a smaller and more reproducible set of associations than MR-Lasso and MR-ALasso, while retaining several biologically plausible relationships supported by external evidence. These findings suggest that MR-ALasso may improve sensitivity for detecting potential causal relationships, whereas MR-ALasso-B provides a more conservative and stable post selection analysis.

\section{Discussion}
\label{sec:discussion}

In this paper, we proposed two Lasso-type procedures for summary-data MR that are robust to pleiotropy. The first method, MR-ALasso, replaces the ordinary lasso penalty with an adaptive lasso penalty in order to improve the correct identification of valid and invalid instruments. The second method, MR-ALasso-B, applies bootstrap smoothing to the post-selection estimator to improve finite-sample inference. The theoretical results, simulation studies, and real-data application suggest that these two procedures address related but distinct aspects of the invalid instrument problem in the two-sample MR framework.

The standard MR-Lasso procedure may provide reasonable point estimation in some settings, but its post-selection inference can be unreliable. In the simulation study, MR-Lasso often had small bias and high power under the alternative, especially in moderate settings. However, under the null, its confidence intervals frequently undercovered, and its type-I error was inflated. This problem became more severe as the proportion of invalid instruments increased and when invalid instruments were stronger than valid instruments. These findings suggest that variable selection alone is not sufficient for valid uncertainty quantification after selection. MR-ALasso improved upon MR-Lasso at the estimation and selection stage. Across a broad range of settings, MR-ALasso generally had smaller bias and smaller RMSE than MR-Lasso, and this advantage was especially clear under directional pleiotropy and in designs where invalid instruments were stronger than valid instruments. The summary-data simulation design also showed that MR-ALasso more accurately identified invalid instruments than MR-Lasso, which is consistent with Proposition \ref{prop:MR-ALasso}. These findings suggest that adaptive penalization improves the separation of variants with stronger evidence of invalidity from those that are closer to valid instruments.

At the same time, improved selection did not by itself resolve the inference problem. Although MR-ALasso improved estimation relative to MR-Lasso, naive post-selection confidence intervals based on the selected model were still too narrow in many settings, and the type-I error remained above the nominal level. This motivates the second contribution of the paper, namely the bootstrap-smoothed estimator MR-ALasso-B. In many of the settings considered, MR-ALasso-B produced much better type-I error control and improved coverage relative to both MR-Lasso and MR-ALasso, especially under the null hypothesis and in more difficult pleiotropic designs. This idea is in line with a broader post-selection inference literature, where bootstrap smoothing has been used to account for the additional variability introduced by data-adaptive selection \citep{efron2014estimation,wang2014discussion}. In the present setting, MR-ALasso-B was particularly useful when reliable interval estimation and hypothesis testing were of primary interest. The gains from MR-ALasso-B should, however, be interpreted with some care. In the main simulation design, the bootstrap-smoothed procedure provided a favorable trade-off between estimation error and inferential reliability, with only moderate losses in power in the harder settings. In the second simulation design, which was introduced to study selection behavior under directly simulated summary-data associations, the results were more mixed. In that setting, MR-ALasso-B often improved coverage, but in the hardest cases with a large proportion of invalid instruments, it could become overly conservative and more biased than MR-ALasso. Thus, the strongest evidence from this paper is that MR-ALasso improves estimation and instrument classification relative to MR-Lasso, whereas the main benefit of MR-ALasso-B lies in improved post-selection inference rather than uniformly improved point estimation.

The present work has some limitations. First, like MR-Lasso and many other selection-based MR procedures, the proposed methods rely on a sparse invalid-instrument framework. If pleiotropic effects are widespread and do not admit a useful separation between approximately valid and invalid instruments, then selection-based procedures may be less effective. Second, the methods depend on tuning choices, including the heterogeneity-based stopping rule and the threshold used to define the aggregated valid set. Third, the bootstrap-smoothed procedure proposed here is not exact selective inference, and should therefore be viewed as a practical inferential correction rather than a full solution to post-selection inference. Fourth, the methods do not explicitly model sample structure, such as population stratification, cryptic relatedness, or sample overlap, which may remain important in some applications \citep{hu2022mendelian}. Finally, although the results in this
paper provide theoretical support for the adaptive penalization step and an Efron-type justification for the bootstrap-smoothed variance estimator, further work is needed to develop a full asymptotic theory for the bootstrap-smoothed estimator as a causal MR procedure. Several extensions would be of interest. One direction is to study alternative aggregation rules for defining the selected valid set, including thresholds motivated by the stability-selection literature. Another is to investigate whether bootstrap smoothing can be combined with other robust MR estimators beyond lasso-based procedures. Extensions to multivariable MR would also be valuable, since instrument selection and post-selection inference are likely to be even more difficult in that setting. A further direction is to develop stronger theoretical results for the asymptotic behavior of MR-ALasso-B under weak instruments, correlated instruments, or sample overlap.

In summary, the results of this study support three main conclusions. First, MR-ALasso improves upon MR-Lasso in estimation accuracy and instrument selection through adaptive penalization. Second, naive post-selection inference remains problematic for both MR-Lasso and MR-ALasso. Third, bootstrap smoothing provides a practical way to improve inference after adaptive-lasso selection. MR-ALasso-B often achieved a favorable combination of improved coverage, better type-I error control, and acceptable estimation error. For applied Mendelian randomization studies in which invalid instruments are a serious concern, MR-ALasso and MR-ALasso-B provide useful additions to the current set of robust MR tools and may serve as useful components of a broader sensitivity analysis in the two-sample MR framework.

\subsection*{Code availability}

The R package \texttt{MRAlasso} is publicly available at
\url{https://github.com/Qasim-stat/MRAlasso}. The version used in this paper
has been archived on Zenodo with DOI:
\href{https://doi.org/10.5281/zenodo.20843253}{10.5281/zenodo.20843253}.

\newpage
\bibliography{references}

@article{kang2016instrumental,
  title={Instrumental variables estimation with some invalid instruments and its application to Mendelian randomization},
  author={Kang, Hyunseung and Zhang, Anru and Cai, T Tony and Small, Dylan S},
  journal={Journal of the American statistical Association},
  volume={111},
  number={513},
  pages={132--144},
  year={2016},
  publisher={Taylor \& Francis}
}

@article{windmeijer2019use,
  title={On the use of the lasso for instrumental variables estimation with some invalid instruments},
  author={Windmeijer, Frank and Farbmacher, Helmut and Davies, Neil and Davey Smith, George},
  journal={Journal of the American Statistical Association},
  volume={114},
  number={527},
  pages={1339--1350},
  year={2019},
  publisher={Taylor \& Francis}
}

@article{qasim2025lasso,
  title={LASSO-type instrumental variable selection methods with an application to Mendelian randomization},
  author={Qasim, Muhammad and M{\aa}nsson, Kristofer and Balakrishnan, Narayanaswamy},
  journal={Statistical Methods in Medical Research},
  volume={34},
  number={2},
  pages={201--223},
  year={2025},
  publisher={Sage Publications Sage UK: London, England}
}

@article{bowden2017framework,
  title={A framework for the investigation of pleiotropy in two-sample summary data Mendelian randomization},
  author={Bowden, Jack and Del Greco M, Fabiola and Minelli, Cosetta and Davey Smith, George and Sheehan, Nuala and Thompson, John},
  journal={Statistics in medicine},
  volume={36},
  number={11},
  pages={1783--1802},
  year={2017},
  publisher={Wiley Online Library}
}

@article{burgess2013mendelian,
  title={Mendelian randomization analysis with multiple genetic variants using summarized data},
  author={Burgess, Stephen and Butterworth, Adam and Thompson, Simon G},
  journal={Genetic epidemiology},
  volume={37},
  number={7},
  pages={658--665},
  year={2013}
  }

@article{bowden2016consistent,
  title={Consistent estimation in Mendelian randomization with some invalid instruments using a weighted median estimator},
  author={Bowden, Jack and Davey Smith, George and Haycock, Philip C and Burgess, Stephen},
  journal={Genetic epidemiology},
  volume={40},
  number={4},
  pages={304--314},
  year={2016},
  publisher={Wiley Online Library}
}

@article{hartwig2017robust,
  title={Robust inference in summary data Mendelian randomization via the zero modal pleiotropy assumption},
  author={Hartwig, Fernando Pires and Davey Smith, George and Bowden, Jack},
  journal={International journal of epidemiology},
  volume={46},
  number={6},
  pages={1985--1998},
  year={2017},
  publisher={Oxford University Press}
}

@article{DaveySmith2003,
  title={‘Mendelian randomization’: can genetic epidemiology contribute to understanding environmental determinants of disease?},
  author={Davey Smith, George and Ebrahim, Shah},
  journal={International journal of epidemiology},
  volume={32},
  number={1},
  pages={1--22},
  year={2003},
  publisher={Oxford University Press}
}

@article{lawlor2008mendelian,
  title={Mendelian randomization: using genes as instruments for making causal inferences in epidemiology},
  author={Lawlor, Debbie A and Harbord, Roger M and Sterne, Jonathan AC and Timpson, Nic and Davey Smith, George},
  journal={Statistics in medicine},
  volume={27},
  number={8},
  pages={1133--1163},
  year={2008},
  publisher={Wiley Online Library}
}

@article{burgess2013use,
  title={Use of allele scores as instrumental variables for Mendelian randomization},
  author={Burgess, Stephen and Thompson, Simon G},
  journal={International journal of epidemiology},
  volume={42},
  number={4},
  pages={1134--1144},
  year={2013},
  publisher={Oxford University Press}
}

@article{Bowden2015,
  author = {Bowden, Jack and Davey Smith, George and Burgess, Stephen},
  title = {Mendelian randomization with invalid instruments: effect estimation and bias detection through Egger regression},
  journal = {International Journal of Epidemiology},
  year = {2015},
  volume = {44},
  number = {2},
  pages = {512--525}
}

@article{Hartwig2017,
  author = {Hartwig, Fernando P. and Davey Smith, George and Bowden, Jack},
  title = {Robust inference in summary data Mendelian randomization via the zero modal pleiotropy assumption},
  journal = {International Journal of Epidemiology},
  year = {2017},
  volume = {46},
  number = {6},
  pages = {1985--1998}
}

@article{rees2019robust,
  title={Robust methods in Mendelian randomization via penalization of heterogeneous causal estimates},
  author={Rees, Jessica MB and Wood, Angela M and Dudbridge, Frank and Burgess, Stephen},
  journal={PloS one},
  volume={14},
  number={9},
  pages={e0222362},
  year={2019},
  publisher={Public Library of Science San Francisco, CA USA}
}

@article{burgess2023guidelines,
  title={Guidelines for performing Mendelian randomization investigations: update for summer 2023},
  author={Burgess, Stephen and Smith, George Davey and Davies, Neil M and Dudbridge, Frank and Gill, Dipender and Glymour, M Maria and Hartwig, Fernando P and Kutalik, Zolt{\'a}n and Holmes, Michael V and Minelli, Cosetta and others},
  journal={Wellcome open research},
  volume={4},
  pages={186},
  year={2023}
}

@article{sanderson2022mendelian,
  title={Mendelian randomization},
  author={Sanderson, Eleanor and Glymour, M Maria and Holmes, Michael V and Kang, Hyunseung and Morrison, Jean and Munaf{\`o}, Marcus R and Palmer, Tom and Schooling, C Mary and Wallace, Chris and Zhao, Qingyuan and others},
  journal={Nature reviews Methods primers},
  volume={2},
  number={1},
  pages={6},
  year={2022},
  publisher={Nature Publishing Group UK London}
}

@article{slob2020comparison,
  title={A comparison of robust Mendelian randomization methods using summary data},
  author={Slob, Eric AW and Burgess, Stephen},
  journal={Genetic epidemiology},
  volume={44},
  number={4},
  pages={313--329},
  year={2020},
  publisher={Wiley Online Library}
}

@article{verbanck2018detection,
  title={Detection of widespread horizontal pleiotropy in causal relationships inferred from Mendelian randomization between complex traits and diseases},
  author={Verbanck, Marie and Chen, Chia-Yen and Neale, Benjamin and Do, Ron},
  journal={Nature genetics},
  volume={50},
  number={5},
  pages={693--698},
  year={2018},
  publisher={Nature Publishing Group US New York}
}

@article{grant2021pleiotropy,
  title={Pleiotropy robust methods for multivariable Mendelian randomization},
  author={Grant, Andrew J and Burgess, Stephen},
  journal={Statistics in medicine},
  volume={40},
  number={26},
  pages={5813--5830},
  year={2021},
  publisher={Wiley Online Library}
}

@article{cheng2015select,
  title={Select the valid and relevant moments: An information-based LASSO for GMM with many moments},
  author={Cheng, Xu and Liao, Zhipeng},
  journal={Journal of Econometrics},
  volume={186},
  number={2},
  pages={443--464},
  year={2015},
  publisher={Elsevier}
}

@article{zou2006adaptive,
  title={The adaptive lasso and its oracle properties},
  author={Zou, Hui},
  journal={Journal of the American statistical association},
  volume={101},
  number={476},
  pages={1418--1429},
  year={2006},
  publisher={Taylor \& Francis}
}

@book{buhlmann2011statistics,
  title={Statistics for high-dimensional data: methods, theory and applications},
  author={B{\"u}hlmann, Peter and Van De Geer, Sara},
  year={2011},
  publisher={Springer Science \& Business Media}
}

@article{belloni2012sparse,
  title={Sparse models and methods for optimal instruments with an application to eminent domain},
  author={Belloni, Alexandre and Chen, Daniel and Chernozhukov, Victor and Hansen, Christian},
  journal={Econometrica},
  volume={80},
  number={6},
  pages={2369--2429},
  year={2012},
  publisher={Wiley Online Library}
}

@article{zhao2020statistical,
  title={Statistical inference in two-sample summary-data Mendelian randomization using robust adjusted profile score},
  author={Zhao, Qingyuan and Wang, Jingshu and Hemani, Gibran and Bowden, Jack and Small, Dylan S},
  journal={Annals of Statistics},
   volume={48},
  pages={1742--1769},
  year={2020},
  publisher={Institute of Mathematical Statistics}
}

@article{patel2024selecting,
  title={Selecting invalid instruments to improve Mendelian randomization with two-sample summary data},
  author={Patel, Ashish and DiTraglia, Francis J and Zuber, Verena and Burgess, Stephen},
  journal={The annals of applied statistics},
  volume={18},
  number={2},
  pages={23--aoas1856},
  year={2024}
}

@article{bowden2019improving,
  title={Improving the accuracy of two-sample summary-data Mendelian randomization: moving beyond the NOME assumption},
  author={Bowden, Jack and Del Greco M, Fabiola and Minelli, Cosetta and Zhao, Qingyuan and Lawlor, Debbie A and Sheehan, Nuala A and Thompson, John and Davey Smith, George},
  journal={International journal of epidemiology},
  volume={48},
  number={3},
  pages={728--742},
  year={2019},
  publisher={Oxford University Press}
}

@article{xue2021constrained,
  title={Constrained maximum likelihood-based Mendelian randomization robust to both correlated and uncorrelated pleiotropic effects},
  author={Xue, Haoran and Shen, Xiaotong and Pan, Wei},
  journal={The American Journal of Human Genetics},
  volume={108},
  number={7},
  pages={1251--1269},
  year={2021},
  publisher={Elsevier}
}

@article{burgess2020robust,
  title={A robust and efficient method for Mendelian randomization with hundreds of genetic variants},
  author={Burgess, Stephen and Foley, Christopher N and Allara, Elias and Staley, James R and Howson, Joanna MM},
  journal={Nature communications},
  volume={11},
  number={1},
  pages={376},
  year={2020},
  publisher={Nature Publishing Group UK London}
}

@article{bowden2018improving,
  title={Improving the visualization, interpretation and analysis of two-sample summary data Mendelian randomization via the Radial plot and Radial regression},
  author={Bowden, Jack and Spiller, Wesley and Del Greco M, Fabiola and Sheehan, Nuala and Thompson, John and Minelli, Cosetta and Davey Smith, George},
  journal={International journal of epidemiology},
  volume={47},
  number={4},
  pages={1264--1278},
  year={2018},
  publisher={Oxford University Press}
}

@article{xie2026winner,
  title={Winner’s Curse Free Robust Mendelian Randomization with Summary Data},
  author={Xie, Zhongming and Zhang, Wanheng and Wang, Jingshen and Wu, Chong},
  journal={Journal of the American Statistical Association},
  pages={1--13},
  year={2026},
  publisher={Taylor \& Francis}
}

@article{meinshausen2010stability,
  title={Stability selection},
  author={Meinshausen, Nicolai and B{\"u}hlmann, Peter},
  journal={Journal of the Royal Statistical Society Series B: Statistical Methodology},
  volume={72},
  number={4},
  pages={417--473},
  year={2010},
  publisher={Oxford University Press}
}

@article{shah2013variable,
  title={Variable selection with error control: another look at stability selection},
  author={Shah, Rajen D and Samworth, Richard J},
  journal={Journal of the Royal Statistical Society Series B: Statistical Methodology},
  volume={75},
  number={1},
  pages={55--80},
  year={2013},
  publisher={Oxford University Press}
}

@inproceedings{bach2008bolasso,
  title={Bolasso: model consistent lasso estimation through the bootstrap},
  author={Bach, Francis R},
  booktitle={Proceedings of the 25th international conference on Machine learning},
  pages={33--40},
  year={2008}
}

@article{wang2014discussion,
  title={Discussion of “Estimation and Accuracy after Model Selection” by Brad Efron},
  author={Wang, Lan and Sherwood, Ben and Li, Runze},
  journal={Journal of the American Statistical Association},
  volume={109},
  number={507},
  pages={1007},
  year={2014}
}

@article{efron2014estimation,
  title={Estimation and accuracy after model selection},
  author={Efron, Bradley},
  journal={Journal of the American Statistical Association},
  volume={109},
  number={507},
  pages={991--1007},
  year={2014},
  publisher={Taylor \& Francis}
}

@article{buja2006observations,
  title={Observations on bagging},
  author={Buja, Andreas and Stuetzle, Werner},
  journal={Statistica Sinica},
  pages={323--351},
  year={2006},
  publisher={JSTOR}
}

@article{hu2022mendelian,
  title={Mendelian randomization for causal inference accounting for pleiotropy and sample structure using genome-wide summary statistics},
  author={Hu, Xianghong and Zhao, Jia and Lin, Zhixiang and Wang, Yang and Peng, Heng and Zhao, Hongyu and Wan, Xiang and Yang, Can},
  journal={Proceedings of the National Academy of Sciences},
  volume={119},
  number={28},
  pages={e2106858119},
  year={2022},
  publisher={National Academy of Sciences}
}

@article{yavorska2017mendelianrandomization,
  title={MendelianRandomization: an R package for performing Mendelian randomization analyses using summarized data},
  author={Yavorska, Olena O and Burgess, Stephen},
  journal={International journal of epidemiology},
  volume={46},
  number={6},
  pages={1734--1739},
  year={2017},
  publisher={Oxford University Press}
}

@software{qasim_mralasso_2026,
  author    = {Qasim, Muhammad},
  title     = {{MRAlasso}: Adaptive Penalization and Bootstrap-Smoothed Inference in Two-Sample Mendelian Randomization},
  year      = {2026},
  publisher = {Zenodo},
  version   = {v0.1.0},
  doi       = {10.5281/zenodo.20843253},
  url       = {https://doi.org/10.5281/zenodo.20843253}
}

@article{broadbent2020mendelianrandomization,
  title   = {{MendelianRandomization} v0.5.0: updates to an R package for performing Mendelian randomization analyses using summarized data},
  author  = {Broadbent, Jim R. and Foley, Christopher N. and Grant, Andrew J. and Mason, Amy M. and Staley, James R. and Burgess, Stephen},
  year    = {2020},
  journal = {Wellcome Open Research},
  doi     = {10.12688/wellcomeopenres.16374.2}
}

@article{diabetes2002reduction,
  title   = {Reduction in the incidence of type 2 diabetes with lifestyle intervention or metformin},
  author  = {Knowler, William C. and Barrett-Connor, Elizabeth and Fowler, Sarah E. and Hamman, Richard F. and Lachin, John M. and Walker, Elizabeth A. and Nathan, David M. and {Diabetes Prevention Program Research Group}},
journal={New England journal of medicine},
  volume={346},
  number={6},
  pages={393--403},
  year={2002},
  publisher={Mass Medical Soc}
}

@article{yang2023effect,
  title={Effect of weight loss on blood pressure changes in overweight patients: a systematic review and meta-analysis},
  author={Yang, Shijie and Zhou, Zhanyang and Miao, Huanhuan and Zhang, Yuqing},
  journal={The Journal of Clinical Hypertension},
  volume={25},
  number={5},
  pages={404--415},
  year={2023},
  publisher={Wiley Online Library}
}

@article{mortensen2023low,
  title={Low-density lipoprotein cholesterol is predominantly associated with atherosclerotic cardiovascular disease events in patients with evidence of coronary atherosclerosis: the Western Denmark Heart Registry},
  author={Mortensen, Martin B{\o}dtker and Dzaye, Omar and B{\o}tker, Hans Erik and Jensen, Jesper M{\o}ller and Maeng, Michael and Bentzon, Jacob Fog and Kanstrup, Helle and S{\o}rensen, Henrik Toft and Leipsic, Jonathon and Blankstein, Ron and others},
  journal={Circulation},
  volume={147},
  number={14},
  pages={1053--1063},
  year={2023},
  publisher={Lippincott Williams \& Wilkins Hagerstown, MD}
}

@article{scandinavian1994randomised,
  title={Randomised trial of cholesterol lowering in 4444 patients with coronary heart disease: the Scandinavian Simvastatin Survival Study (4S)},
  author={{Scandinavian Simvastatin Survival Study Group and others}},
  journal={The Lancet},
  volume={344},
  number={8934},
  pages={1383--1389},
  year={1994},
  publisher={Elsevier}
}

@article{ettehad2016blood,
  title={Blood pressure lowering for prevention of cardiovascular disease and death: a systematic review and meta-analysis},
  author={Ettehad, Dena and Emdin, Connor A and Kiran, Amit and Anderson, Simon G and Callender, Thomas and Emberson, Jonathan and Chalmers, John and Rodgers, Anthony and Rahimi, Kazem},
  journal={The Lancet},
  volume={387},
  number={10022},
  pages={957--967},
  year={2016},
  publisher={Elsevier}
}

@article{dudbridge2025getting,
  title={Getting to GRIPS with MR-Egger: Modelling directional pleiotropy independently of allele coding},
  author={Dudbridge, Frank and Voller, Bethany and Woodward, Ruby M and Saxby, Katie L and Frayling, Timothy M and Pilling, Luke C and Bowden, Jack},
  journal={PLoS genetics},
  volume={21},
  number={12},
  pages={e1011967},
  year={2025},
  publisher={Public Library of Science San Francisco, CA USA}
}

\newpage
\appendix
\renewcommand{\theequation}{A\arabic{equation}}
\setcounter{equation}{0}

\renewcommand{\thetable}{A\arabic{table}}
\setcounter{table}{0}

\section*{Supplementary Materials}
\section{MR-ALasso procedure} \label{app:Alasso}

The objective function in \eqref{eq:Alasso} is a weighted least-squares criterion with partial penalization: the causal effect parameter $\theta$ is left unpenalized, whereas the genetic variant-specific direct effects $\alpha_j$ are penalized. For computational and theoretical purposes, it is useful to rewrite the criterion in matrix form. Suppose 
$$\hat{\boldsymbol{\beta}}_Y =
(\hat{\beta}_{Y1},\dots,\hat{\beta}_{YJ})^\top,\qquad\hat{\boldsymbol{\beta}}_X = (\hat{\beta}_{X1},\dots,\hat{\beta}_{XJ})^\top,
\qquad \bm S = \mathrm{diag}(\tilde{w}_1,\dots,\tilde{w}_J),$$
and denote by $\boldsymbol{\alpha} = (\alpha_1,\dots,\alpha_J)^\top$ the vector of direct effects. Then \eqref{eq:Alasso} can be written as
\begin{equation}
(\hat{\theta}_{\mathrm{AL}},\hat{\boldsymbol{\alpha}}_{\mathrm{AL}})=
\arg\min_{\theta,\boldsymbol{\alpha}}\frac{1}{2}
\left\|\bm S^{1/2}\left(
\hat{\boldsymbol{\beta}}_Y - \theta \hat{\boldsymbol{\beta}}_X - \boldsymbol{\alpha}\right)\right\|_2^2+
\lambda_n\sum_{j=1}^{J}\omega_j |\alpha_j|.
\label{eq:Alasso_matrix}
\end{equation}
This is not a usual adaptive Lasso problem because not all regression parameters are penalized. However, as in the instrumental variables formulation of \citet{windmeijer2019use}, the unpenalized parameter can be removed by projection. To see this, define
$$\widetilde{\mathbf{b}} = \bm S^{1/2}\hat{\boldsymbol{\beta}}_X,
\qquad
P_{\widetilde{\mathbf{b}}} =\widetilde{\mathbf{b}}(\widetilde{\mathbf{b}}^\top \widetilde{\mathbf{b}})^{-1}\widetilde{\mathbf{b}}^\top,
\qquad
M_{\widetilde{\mathbf{b}}} = \bm I_J - P_{\widetilde{\mathbf{b}}},$$
where $P_{\widetilde{\mathbf{b}}}$ is the projection matrix onto the span of $\widetilde{\mathbf{b}}$ and $M_{\widetilde{\mathbf{b}}}$ is the corresponding residual-maker matrix. Premultiplying the weighted model by $M_{\widetilde{\mathbf{b}}}$ removes the component along $\hat{\boldsymbol{\beta}}_X$, since $M_{\widetilde{\mathbf{b}}} \widetilde{\mathbf{b}} = \mathbf{0}.$
Hence
$$M_{\widetilde{\mathbf{b}}} \bm S^{1/2}\hat{\boldsymbol{\beta}}_Y=
M_{\widetilde{\mathbf{b}}} \bm S^{1/2}\boldsymbol{\alpha}+
M_{\widetilde{\mathbf{b}}} \bm S^{1/2}\boldsymbol{\varepsilon}.$$
If we define
$$\widetilde{\mathbf{y}}_{\mathrm{AL}} = M_{\widetilde{\mathbf{b}}} \bm S^{1/2}\hat{\boldsymbol{\beta}}_Y, \qquad
\widetilde{\mathbf{X}}_{\mathrm{AL}} = M_{\widetilde{\mathbf{b}}} \bm S^{1/2},$$
then \eqref{eq:Alasso_matrix} is equivalent to the adaptive Lasso problem
\begin{equation}
\hat{\boldsymbol{\alpha}}_{\mathrm{AL}}=
\arg\min_{\boldsymbol{\alpha}}\frac{1}{2}
\left\|
\widetilde{\mathbf{y}}_{\mathrm{AL}} - \widetilde{\mathbf{X}}_{\mathrm{AL}} \boldsymbol{\alpha}
\right\|_2^2+\lambda_n\sum_{j=1}^{J}\omega_j |\alpha_j|,
\label{eq:adalasso_projected}
\end{equation}
after which the causal effect can be recovered from the fitted direct effects. 

\section{Proofs} \label{app:proofs}
\subsection{Proof of Theorem \ref{thm:reducedcriterion}} \label{proof:prop_reducedcrit}

\begin{proof}
 From \eqref{eq:Alasso_matrix}, we can write 
 $$Q(\theta,\bm{\alpha})=
\frac{1}{2}\left\|
\bm{u}(\bm{\alpha})-\theta \widetilde{\bm{b}}\right\|_2^2+
\lambda_n \sum_{j=1}^J \omega_{j} |\alpha_j|,$$
where $\bm{u}(\bm{\alpha})=\bm{S}^{1/2} \bigl(\hat{\bm{\beta}}_Y-\bm{\alpha}\bigr)$. Then for fixed $\bm{\alpha}$, the first term of the above expression depends on $\theta$, and we differentiate w.r.t $\theta$ and setting the derivative to zero gives
$$\hat{\theta}(\bm{\alpha})=
\frac{\widetilde{\bm{b}}^\top \bm{u}(\bm{\alpha})
}{\widetilde{\bm{b}}^\top \widetilde{\bm{b}}
}.$$
Substituting $\bm{u}(\bm{\alpha})$ and $\widetilde{\bm{b}}$ in the above expression, we obtain
$$\hat{\theta}(\bm{\alpha})=
\bigl(
\hat{\bm{\beta}}_X^\top \bm{S}\hat{\bm{\beta}}_X\bigr)^{-1}
\hat{\bm{\beta}}_X^\top \bm{S}\bigl(\hat{\bm{\beta}}_Y-\bm{\alpha} \bigr).$$
So the minimization criterion is
$$Q(\bm{\alpha})=\min_{\theta} Q(\theta,\bm{\alpha})=
\frac{1}{2}\left\|M_{\widetilde{\mathbf{b}}}\bm{u}(\bm{\alpha})\right\|_2^2
+\lambda_n \sum_{j=1}^J \omega_{j} |\alpha_j|,$$
where $M_{\widetilde{\mathbf{b}}} = I_J - P_{\widetilde{\mathbf{b}}}$ is symmetric and idempotent matrix, and therefore
$$\left\|M_{\widetilde{\mathbf{b}}}\bm{u}(\bm{\alpha})\right\|_2^2
= \bm{u}(\bm{\alpha})^\top M_{\widetilde{\mathbf{b}}}
\bm{u}(\bm{\alpha}) =\bigl( \hat{\bm{\beta}}_Y-\bm{\alpha} \bigr)^\top \bm{S}^{1/2}
M_{\widetilde{\mathbf{b}}}
\bm{S}^{1/2}\bigl(
\hat{\bm{\beta}}_Y-\bm{\alpha} \bigr).$$
Using $\bm{S}^{1/2}P_{\widetilde{\mathbf{b}}}\bm{S}^{1/2}=\bm{S}\hat{\bm{\beta}}_X
\bigl(\hat{\bm{\beta}}_X^\top \bm{S}\hat{\bm{\beta}}_X
\bigr)^{-1}
\hat{\bm{\beta}}_X^\top \bm{S},$ 
we get
$$\bm{S}^{1/2}M_{\widetilde{\mathbf{b}}}\bm{S}^{1/2}=
\bm{S}-\bm{S}\hat{\bm{\beta}}_X
\bigl(
\hat{\bm{\beta}}_X^\top \bm{S}\hat{\bm{\beta}}_X
\bigr)^{-1}
\hat{\bm{\beta}}_X^\top \bm{S}=
\hat{\bm{C}}_n.$$
Hence
$$Q(\bm{\alpha})=
L_n(\boldsymbol\alpha)+\lambda_n \sum_{j=1}^J \omega_{j} |\alpha_j|,
$$
with 
$$L_n(\boldsymbol\alpha)=
\frac{1}{2}
(\hat{\boldsymbol\beta}_Y-\boldsymbol\alpha)^\top
\hat{\mathbf C}_n
(\hat{\boldsymbol\beta}_Y-\boldsymbol\alpha).$$
This is Eq. \eqref{eq:Alasso-alpha} in the main text. Evaluating $\hat{\theta}(\bm{\alpha})$ at $\bm{\alpha}=\hat{\bm{\alpha}}$ gives the corresponding estimator of the causal effect. 
\end{proof}

\subsection{Proof of Proposition \ref{prop:MR-ALasso}} \label{proof:selectionconsist}

\begin{proof}
Define the initial estimator as $\tilde{\boldsymbol\alpha} = \hat{\boldsymbol\beta}_Y-\tilde{\theta}\hat{\boldsymbol\beta}_X$. Under the population summary-data model,
$$\boldsymbol\beta_Y=\theta\boldsymbol\beta_X+\boldsymbol\alpha.$$
Subtracting $\boldsymbol\alpha$,
\begin{align}
\tilde{\boldsymbol\alpha}-\boldsymbol\alpha
&=
\hat{\boldsymbol\beta}_Y-\tilde{\theta}\hat{\boldsymbol\beta}_X-
(\boldsymbol\beta_Y-\theta\boldsymbol\beta_X)
\nonumber\\
&=
(\hat{\boldsymbol\beta}_Y-\boldsymbol\beta_Y)-
(\tilde{\theta}-\theta)\hat{\boldsymbol\beta}_X-
\theta(\hat{\boldsymbol\beta}_X-\boldsymbol\beta_X).
\label{eq:alpha_minus_atilde}
\end{align}

By Assumption \ref{assump:two-sample_summary_data}, the summary associations, $\sqrt n(\hat{\boldsymbol\beta}_Y-\boldsymbol\beta_Y)=O_p(1)$ and  $\sqrt n(\hat{\boldsymbol\beta}_X-\boldsymbol\beta_X)=O_p(1)$. Moreover, we consider an initial consistent estimator by following the adaptive Lasso procedure of \cite{zou2006adaptive}. We build on the result of \cite{windmeijer2019use}, they show that the median of $\hat{\beta}_{Yj}/\hat{\beta}_{Xj}$, denoted  $\tilde{\theta}$, is a consistent estimator for $\theta$ under Assumption \ref{assump:pleiotropy_structure}(i). Since \(\hat{\boldsymbol\beta}_X=O_p(1)\), it follows from \eqref{eq:alpha_minus_atilde} that $\sqrt n(\tilde{\boldsymbol\alpha}-\boldsymbol\alpha)=O_p(1)$.

Now consider the coordinates separately. If $j\in\mathcal A$, then $\alpha_j\neq 0$, so $\tilde{\alpha}_j\xrightarrow{p}\alpha_j\neq 0.$ If $j\in\mathcal V$, then $\alpha_j=0$, and therefore $\tilde{\alpha}_j=O_p(n^{-1/2}).$ From these results, we show the asymptotic behavior of the adaptive weights
\begin{equation}
\omega_{j}=|\tilde{\alpha}_j|^{-\nu}\xrightarrow{p}|\alpha_j|^{-\nu}, \qquad j\in\mathcal A,
\label{eq:weights_active_op}
\end{equation}
and 
\begin{equation}
\omega_{j}=|\tilde{\alpha}_j|^{-\nu}=O_p(n^{\nu/2}), \qquad  j\in\mathcal V.
\label{eq:weights_inactive}
\end{equation}
Thus the adaptive weights remain bounded on truly invalid variants and diverge on truly valid variants.
For $j\in\mathcal A$, combining \eqref{eq:weights_active_op} with \eqref{eq:tuning_prop} yields $\lambda_n\omega_{j} = \lambda_n O_p(1),$ so
\begin{equation}
\frac{\lambda_n\omega_{j}}{\sqrt n}\xrightarrow{p}0,
\qquad j\in\mathcal A.
\label{eq:penalty_active_prop}
\end{equation}
On the other hand, for $j\in\mathcal V$, using \eqref{eq:weights_inactive}, $\lambda_n\omega_{j}=\lambda_n O_p(n^{\nu/2}),$ and therefore ${\lambda_n\omega_{j}}/{\sqrt n}=\lambda_n O_p(n^{(\nu-1)/2}).$
By \eqref{eq:tuning_prop},
\begin{equation}
\frac{\lambda_n\omega_{j}}{\sqrt n}\xrightarrow{p}\infty,
\qquad j\in\mathcal V.
\label{eq:penalty_inactive_prop}
\end{equation}
Thus, on the inactive coordinates, the penalty term dominates the local score contribution from the loss, $L_n(\boldsymbol\alpha)$.

We now show the selection consistency of MR-ALasso. Eq. \eqref{eq:penalty_inactive_prop} implies that the adaptive penalty is asymptotically dominant relative to the local score of the loss function, $L_n(\boldsymbol\alpha)$. Therefore, by the KKT conditions and the standard ALasso argument \citep{zou2006adaptive} 
\begin{equation}
P\bigl(\hat{\alpha}_{\mathrm{AL},j}=0 \ \forall j\in\mathcal V\bigr)\to 1.
\label{eq:no_false_inclusions}
\end{equation}

On the active set $\mathcal A$, the penalty is asymptotically negligible by \eqref{eq:penalty_active_prop}, so the MR-ALasso estimator is asymptotically equivalent to the oracle estimator based on the true active set. Because the weights in $\mathbf S$ are positive, Assumption \ref{assump:pleiotropy_structure}(iii) implies $\beta_{Xj}\neq0$ for all $j$, and $|\mathcal A|<J$. Therefore, the oracle submatrix $\mathbf C_{\mathcal A\mathcal A}$ is positive definite. Consequently, the oracle loss restricted to the active coordinates is locally strictly convex, and its sample minimizer is consistent for $\boldsymbol\alpha_{\mathcal A}$. By Assumption \ref{assump:pleiotropy_structure}(ii), which ensures that $\min_{j\in\mathcal A}|\alpha_j|\ge c_\alpha>0,$ it follows that 
\begin{equation}
P\bigl(\hat{\alpha}_{\mathrm{AL},j}\neq 0 \ \forall j\in\mathcal A\bigr)\to 1.
\label{eq:no_false_exclusions}
\end{equation}

Combining \eqref{eq:no_false_inclusions} and \eqref{eq:no_false_exclusions}, we obtain
$$P(\widehat{\mathcal A}_n=\mathcal A)\to 1.$$
\end{proof}

\subsection{Proof of Corollary \ref{cor:oracle_equiv}} \label{proof:oracle_equiv}
\begin{proof}
Let $\widehat{\mathcal V}_n=\{j:\hat{\alpha}_{\mathrm{AL},j}=0\}$ be the set of variants as valid by MR-ALasso. Define the post-MR-ALasso estimator by $\hat{\theta}_{\mathrm{post\mbox{-}AL}}=\underset{\theta}{\arg\min} \sum_{j\in \widehat{\mathcal V}_n}\tilde{w}_j \bigl(\hat{\beta}_{Yj}-\theta \hat{\beta}_{Xj}\bigr)^2$, and let $\hat{\theta}_{\mathrm{oracle}}$ denote the the oracle IVW estimator based on the true valid set $\mathcal V$.
By Proposition \ref{prop:MR-ALasso}, $P(\widehat{\mathcal A}_n=\mathcal A)\to 1.$ Since $\widehat{\mathcal V}_n=\widehat{\mathcal A}_n^c$ and $\mathcal V=\mathcal A^c$, it follows that $P(\widehat{\mathcal V}_n=\mathcal V)\to 1.$ On the event $\{\widehat{\mathcal V}_n=\mathcal V\}$, then $\hat{\theta}_{\mathrm{post\mbox{-}AL}}$ and $\hat{\theta}_{\mathrm{oracle}}$ are computed from the same set of variants and therefore are identical. Hence $P\!\left( \hat{\theta}_{\mathrm{post\mbox{-}AL}} = \hat{\theta}_{\mathrm{oracle}} \right)\to 1,$ which establishes the asymptotic equivalence of the two estimators.
\end{proof}

\subsection{Proof of Theorem \ref{thm:mralassob_efron_delta}}
\label{supp:mralassob_efron_delta}
\begin{proof}
The proof follows the nonparametric delta-method argument of \citet{efron2014estimation}, adapted to the SNP-level bootstrap and to the set $\mathcal R$.

For $j=1,\ldots,J$, consider the one-dimensional perturbation
$$p_\epsilon=p_0+\epsilon(\delta_j-p_0),$$
where $\delta_j$ is the $j$th coordinate vector. Thus,
$$
p_{\epsilon j}=\frac{1}{J}+\epsilon\left(1-\frac{1}{J}\right),
\qquad
p_{\epsilon k}=\frac{1}{J}(1-\epsilon),\quad k\neq j.
$$
For the purposes of this proof, we define $$\hat\theta_{\mathcal R}(\mathbf w)=
\begin{cases}
\hat\theta(\mathbf w), & \mathbf w\in\mathcal R,\\
0, & \mathbf w\notin\mathcal R.
\end{cases}$$
This convention allows expectations to be written over the full bootstrap count space $\mathcal W_J$, while the original estimator on the set $\mathcal R$.

To simplify the notation, let
$$ A(p) =E_p\{\hat\theta_{\mathcal R}(\mathbf W^*(p))\}= \sum_{\mathbf w\in\mathcal R} \hat\theta(\mathbf w) P_p\{\mathbf W^*(p)=\mathbf w\}$$
and
$$ B(p) = P_p\{\mathbf W^*(p)\in\mathcal R\}.$$
Then $$S_{\mathcal R}(p)=\frac{A(p)}{B(p)}.$$ 

For a fixed count vector $\mathbf w\in\mathcal W_J$, the likelihood ratio of the multinomial probability under $p_\epsilon$ relative to that under $p_0$ is
$$
L_\epsilon(\mathbf w) = \prod_{k=1}^J \left( \frac{p_{\epsilon k}}{1/J} \right)^{w_k}.
$$
Using the form of $p_\epsilon$, this becomes
$$
L_\epsilon(\mathbf w) = \{1+(J-1)\epsilon\}^{w_j} (1-\epsilon)^{\sum_{k\neq j}w_k}.
$$
Since $\sum_{k=1}^J w_k=J$, differentiating at $\epsilon=0$ gives
$$
\left. \frac{d}{d\epsilon} L_\epsilon(\mathbf w) \right|_{\epsilon=0} = J(w_j-1).
$$
Therefore,
$$
\left. \frac{d}{d\epsilon} A(p_\epsilon) \right|_{\epsilon=0} = J E_{p_0} \left[ \hat\theta_{\mathcal R}(\mathbf W^*) (W_j^*-1) \right],
$$
and
$$
\left.
\frac{d}{d\epsilon}B(p_\epsilon)\right|_{\epsilon=0}=
J E_{p_0} \left[ I(\mathbf W^*\in\mathcal R) (W_j^*-1)
\right].
$$
Using the quotient rule,
$
\dot S_{\mathcal R,j}= \left. \frac{d}{d\epsilon} S_{\mathcal R}(p_\epsilon) \right|_{\epsilon=0}
$
is
$$
\dot S_{\mathcal R,j} = J \left\{ E_{p_0} \left[ \hat \theta(\mathbf W^*)(W_j^*-1) \mid
\mathbf W^*\in\mathcal R \right] - E_{p_0} \left[ \hat \theta (\mathbf W^*) \mid
\mathbf W^*\in\mathcal R \right] E_{p_0} \left[ W_j^*-1 \mid
\mathbf W^*\in\mathcal R \right]
\right\}.
$$
Since subtracting a constant does not change covariance,
$$
\dot S_{\mathcal R,j} = J \mathrm{Cov}_{p_0} \left[ W_j^*, \hat \theta (\mathbf W^*) \mid
\mathbf W^*\in\mathcal R \right] = J S_j^*.
$$

The nonparametric delta-method variance for a smooth functional of the empirical instrument distribution assigns variance
$$
\frac{1}{J^2} \sum_{j=1}^J \dot S_{\mathcal R,j}^2.
$$
Substituting \(\dot S_{\mathcal R,j}=J S_j^*\) gives
$$
\frac{1}{J^2} \sum_{j=1}^J (J S_j^*)^2 = \sum_{j=1}^J (S_j^*)^2.
$$
This proves the theorem.
\end{proof}
\subsection{Proof of Proposition \ref{prop:efron_variance_limit}}
\label{supp:variance_limit}
\begin{proof} 
For each SNP $j$, recall that
$\bar W_j^*=\frac{1}{B}\sum_{b=1}^BW_{jb}^*$ and
$\widehat S_j=\frac{1}{B}\sum_{b=1}^B \left(W_{jb}^*-\bar W_j^*\right) \left\{ \hat\theta(\mathbf W_b^*)-\tilde\theta_{\mathrm{AL\mbox{-}B}} \right\}.
$
Expanding this expression gives
$$
\widehat S_j=\frac{1}{B}\sum_{b=1}^BW_{jb}^*\hat\theta(\mathbf W_b^*)-\bar W_j^* \tilde\theta_{\mathrm{AL\mbox{-}B}}.
$$
Conditional on $\mathcal R$, the pairs $\left(W_{jb}^* \hat\theta(\mathbf W_b^*)\right)$, are i.i.d from the bootstrap distribution conditional on $\mathbf W^*\in\mathcal R$. Therefore, by the law of large numbers,
$$
\frac{1}{B}\sum_{b=1}^BW_{jb}^*\xrightarrow{P_*}E_*\left[W_j^*\mid\mathcal R\right],
$$
$$
\frac{1}{B}\sum_{b=1}^B\hat\theta(\mathbf W_b^*)\xrightarrow{P_*}E_*\left[\hat\theta(\mathbf W^*)\mid\mathcal R\right],
$$
and
$$
\frac{1}{B}\sum_{b=1}^BW_{jb}^*\hat\theta(\mathbf W_b^*)\xrightarrow{P_*}E_*\left[W_j^*\hat\theta(\mathbf W^*)\mid\mathcal R\right].
$$
Combining these limits gives
$$
\widehat S_j\xrightarrow{P_*}E_*\left[W_j^*\hat\theta(\mathbf W^*)\mid\mathcal R\right]-E_*\left[W_j^*\mid\mathcal R\right]E_*\left[\hat\theta(\mathbf W^*)\mid\mathcal R\right].
$$
The right-hand side is $\mathrm{Cov}_*\left[W_j^*,\hat\theta(\mathbf W^*)\mid\mathcal R\right]=S_j^*.$
Thus, for each $j=1,\ldots,J$,
$$
\widehat S_j
\xrightarrow{P_*}
S_j^*.
$$
Since $J$ is fixed, the continuous mapping theorem yields
$$
\sum_{j=1}^J
\widehat S_j^2
\xrightarrow{P_*}
\sum_{j=1}^J
(S_j^*)^2.
$$
This proves Proposition \ref{prop:efron_variance_limit}.
\end{proof}
\section{Additional Simulation Results} \label{app:Simulation}

The additional simulation results in Tables \ref{tab:model-I-theta0} and \ref{tab:modelI_theta0.2_stronger_invalid} were consistent with the main findings reported in the manuscript. Under the null hypothesis in the equal strength design, MR-ALasso generally performed better as compared to MR-Lasso in terms of SD and RMSE, but both lasso-based procedures remained anti-conservative. MR-ALasso-B showed much better type-I error control than MR-Lasso and MR-ALasso across the four pleiotropy scenarios. When the true causal effect was positive for a design with invalid instruments that were stronger than the valid ones, MR-ALasso reduced bias and RMSE relative to MR-Lasso, especially in the challenging settings with directional pleiotropy and violation of the InSIDE assumption. MR-ALasso-B provided higher coverage probabilities than the naive lasso-based procedures. In summary, these additional results support conclusions discussed in Section \ref{sec:simulation}. MR-ALasso improved estimation, whereas MR-ALasso-B improved post-selection statistical inference. The complete Model-II simulation results for all data-generating scenarios considered are reported in Tables \ref{tab:modelII_theta0.0} and \ref{tab:modelII_theta0.2}. The performance of the MR-ALasso method was similar to that of oracle IVW in terms of bias and RMSE, providing finite-sample support for Corollary \ref{cor:oracle_equiv}, which states that the post MR-ALasso method is asymptotically identical to the oracle IVW estimator when the invalid set is correctly identified. When the proportion of invalid instruments was smaller, both MR-Lasso and MR-ALasso performed well, although MR-ALasso was slightly more stable. When $p_I$ increased to 45\%, the problem became substantially more difficult, especially in the invalid stronger design. In that case, MR-Lasso often removed many more instruments than were truly invalid, whereas MR-ALasso consistently showed better selection and estimation performance, particularly at the larger sample size. MR-ALasso-B often became overly conservative and could show bias when the proportion of invalid instruments was very high. Model-II supports the conclusion that adaptive penalization improves invalid instrument identification and yields a post-selection estimator that can closely approximate the oracle IVW benchmark under the conditions of Proposition \ref{prop:MR-ALasso}.  

\begin{table}[H]
\caption{Simulation results for Model-I under different pleiotropy scenarios when $\theta=0$ and $\gamma_{\mathcal{A}} = \gamma_{\mathcal{V}}$.}
\centering
\resizebox{\textwidth}{!}{
\begin{tabular}{lcccc cccc cccc}
\toprule
\multirow{3}{*}{Method}
& \multicolumn{4}{c}{15\% Invalid}
& \multicolumn{4}{c}{30\% Invalid}
& \multicolumn{4}{c}{45\% Invalid} \\
\cmidrule(lr){2-5} \cmidrule(lr){6-9} \cmidrule(lr){10-13}
& $\hat{\theta}$ & SD & RMSE & Type-I
& $\hat{\theta}$ & SD & RMSE & Type-I
& $\hat{\theta}$ & SD & RMSE & Type-I \\
\midrule
\multicolumn{13}{l}{\textit{Scenario 1: Balanced pleiotropy, InSIDE satisfied}} \\
MR-IVW        & 0.000 & 0.146 & 0.146 & 0.06 & 0.002 & 0.206 & 0.206 & 0.07 & -0.002 & 0.249 & 0.249 & 0.07 \\
MR-Robust  & 0.000 & 0.034 & 0.034 & 0.05 & 0.000 & 0.046 & 0.046 & 0.05 & 0.002  & 0.084 & 0.084 & 0.02 \\
MR-Egger   & -0.013 & 0.428 & 0.428 & 0.06 & 0.003 & 0.609 & 0.609 & 0.06 & -0.009 & 0.751 & 0.751 & 0.06 \\
MR-Median     & 0.000 & 0.044 & 0.044 & 0.04 & 0.001 & 0.055 & 0.055 & 0.06 & 0.001  & 0.072 & 0.072 & 0.08 \\
MR-Mode       & 0.004 & 0.159 & 0.159 & 0.00 & 0.006 & 0.297 & 0.297 & 0.00 & -0.004 & 0.258 & 0.258 & 0.01 \\
MR-Lasso   & 0.000 & 0.035 & 0.035 & 0.08 & -0.001 & 0.043 & 0.043 & 0.10 & -0.001 & 0.055 & 0.055 & 0.12 \\
MR-ALasso  & 0.000 & 0.035 & 0.035 & 0.07 & 0.000 & 0.042 & 0.042 & 0.09 & 0.000  & 0.054 & 0.054 & 0.11 \\
MR-ALasso-B   & 0.000 & 0.035 & 0.035 & 0.03 & -0.001 & 0.041 & 0.041 & 0.03 & -0.001 & 0.053 & 0.053 & 0.03 \\

\midrule
\multicolumn{13}{l}{\textit{Scenario 2: Directional pleiotropy, InSIDE satisfied}} \\
MR-IVW        & 0.251 & 0.149 & 0.291 & 0.36 & 0.501 & 0.209 & 0.543 & 0.69 & 0.746 & 0.253 & 0.787 & 0.85 \\
MR-Robust  & 0.002 & 0.035 & 0.035 & 0.05 & 0.013 & 0.046 & 0.048 & 0.06 & 0.081 & 0.091 & 0.122 & 0.04 \\
MR-Egger   & -0.003 & 0.474 & 0.474 & 0.06 & 0.014 & 0.658 & 0.658 & 0.05 & 0.021 & 0.807 & 0.807 & 0.06 \\
MR-Median     & 0.022 & 0.045 & 0.050 & 0.06 & 0.055 & 0.056 & 0.079 & 0.17 & 0.111 & 0.080 & 0.137 & 0.39 \\
MR-Mode       & 0.007 & 0.212 & 0.212 & 0.01 & -0.002 & 0.323 & 0.323 & 0.01 & 0.001 & 0.208 & 0.208 & 0.01 \\
MR-Lasso   & 0.005 & 0.036 & 0.037 & 0.09 & 0.018 & 0.044 & 0.048 & 0.13 & 0.055 & 0.064 & 0.084 & 0.30 \\
MR-ALasso  & 0.003 & 0.035 & 0.035 & 0.07 & 0.007 & 0.043 & 0.044 & 0.09 & 0.024 & 0.057 & 0.062 & 0.16 \\
MR-ALasso-B   & 0.002 & 0.035 & 0.035 & 0.03 & 0.009 & 0.042 & 0.043 & 0.03 & 0.025 & 0.056 & 0.062 & 0.03 \\

\midrule
\multicolumn{13}{l}{\textit{Scenario 3: Directional pleiotropy, InSIDE violated}} \\
MR-IVW        & 0.544 & 0.202 & 0.580 & 0.94 & 0.843 & 0.221 & 0.872 & 0.99 & 1.029 & 0.216 & 1.051 & 1.00 \\
MR-Robust  & 0.007 & 0.038 & 0.039 & 0.06 & 0.034 & 0.064 & 0.072 & 0.05 & 0.225 & 0.171 & 0.282 & 0.10 \\
MR-Egger   & 1.351 & 0.565 & 1.464 & 0.91 & 1.610 & 0.517 & 1.691 & 0.96 & 1.638 & 0.461 & 1.702 & 0.97 \\
MR-Median     & 0.125 & 0.078 & 0.148 & 0.62 & 0.369 & 0.221 & 0.430 & 0.94 & 0.703 & 0.303 & 0.765 & 0.99 \\
MR-Mode       & 0.004 & 0.159 & 0.159 & 0.02 & 0.013 & 0.089 & 0.090 & 0.04 & 0.046 & 0.249 & 0.253 & 0.13 \\
MR-Lasso   & 0.020 & 0.040 & 0.045 & 0.17 & 0.092 & 0.070 & 0.116 & 0.61 & 0.284 & 0.154 & 0.323 & 0.95 \\
MR-ALasso  & 0.007 & 0.038 & 0.039 & 0.11 & 0.022 & 0.051 & 0.056 & 0.20 & 0.070 & 0.073 & 0.101 & 0.43 \\
MR-ALasso-B   & 0.007 & 0.037 & 0.037 & 0.03 & 0.025 & 0.048 & 0.054 & 0.04 & 0.079 & 0.069 & 0.105 & 0.06 \\

\midrule
\multicolumn{13}{l}{\textit{Scenario 4: Balanced pleiotropy, InSIDE violated}} \\
MR-IVW        & 0.039 & 0.156 & 0.161 & 0.09 & 0.071 & 0.213 & 0.225 & 0.11 & 0.098 & 0.246 & 0.265 & 0.10 \\
MR-Robust  & 0.002 & 0.035 & 0.035 & 0.05 & 0.002 & 0.048 & 0.048 & 0.05 & 0.014 & 0.091 & 0.092 & 0.03 \\
MR-Egger   & 0.191 & 0.572 & 0.603 & 0.23 & 0.302 & 0.654 & 0.720 & 0.20 & 0.393 & 0.678 & 0.784 & 0.20 \\
MR-Median     & 0.008 & 0.046 & 0.047 & 0.05 & 0.016 & 0.064 & 0.066 & 0.11 & 0.028 & 0.086 & 0.091 & 0.17 \\
MR-Mode       & 0.008 & 0.408 & 0.408 & 0.00 & 0.031 & 2.353 & 2.353 & 0.00 & -0.020 & 1.062 & 1.062 & 0.00 \\
MR-Lasso   & 0.002 & 0.035 & 0.035 & 0.08 & 0.003 & 0.045 & 0.045 & 0.11 & 0.007 & 0.056 & 0.057 & 0.16 \\
MR-ALasso  & 0.002 & 0.035 & 0.035 & 0.07 & 0.000 & 0.044 & 0.044 & 0.10 & 0.001 & 0.054 & 0.054 & 0.13 \\
MR-ALasso-B   & 0.001 & 0.034 & 0.034 & 0.03 & 0.000 & 0.043 & 0.043 & 0.03 & 0.002 & 0.052 & 0.052 & 0.03 \\

\bottomrule
\end{tabular}
}
\begin{tablenotes}[flushleft]
\scriptsize
\item Note: $\hat{\theta}$ is the empirical mean of the estimated causal effect; SD is the standard deviation; RMSE is the root mean squared error; Type-I error is the empirical rejection probability under the null hypothesis $\theta=0$; and $\gamma_{\mathcal{A}} = \gamma_{\mathcal{V}}$ indicates that $E(\gamma_j \mid j\in\mathcal{A}) =
E(\gamma_j \mid j\in\mathcal{V})$, so that valid and invalid variants have equal instrument strength.
\end{tablenotes}
\label{tab:model-I-theta0}
\end{table}

\begin{table}[H]
\centering
\caption{Simulation results for Model-I under different pleiotropy scenarios when $\theta=0.20$ and $\gamma_{\mathcal{A}} = 3\gamma_{\mathcal{V}}$.}
\resizebox{\textwidth}{!}{
\begin{tabular}{lccccc ccccc ccccc}
\toprule
\multirow{2}{*}{Method}
& \multicolumn{5}{c}{15\% Invalid}
& \multicolumn{5}{c}{30\% Invalid}
& \multicolumn{5}{c}{45\% Invalid} \\
\cmidrule(lr){2-6} \cmidrule(lr){7-11} \cmidrule(lr){12-16}
& $\hat{\theta}$ & SD & RMSE & CP & Power
& $\hat{\theta}$ & SD & RMSE & CP & Power
& $\hat{\theta}$ & SD & RMSE & CP & Power \\
\midrule

\multicolumn{16}{l}{\textit{Scenario 1: Balanced pleiotropy, InSIDE satisfied}} \\

MR-IVW      & 0.185 & 0.280 & 0.280 & 0.65 & 0.44 & 0.190 & 0.323 & 0.323 & 0.75 & 0.32 & 0.183 & 0.338 & 0.339 & 0.82 & 0.23 \\
MR-Robust   & 0.169 & 0.070 & 0.077 & 0.91 & 0.67 & 0.160 & 0.149 & 0.155 & 0.91 & 0.24 & 0.168 & 0.303 & 0.304 & 0.92 & 0.14 \\
MR-Egger    & 0.161 & 0.539 & 0.540 & 0.61 & 0.40 & 0.175 & 0.527 & 0.528 & 0.79 & 0.23 & 0.162 & 0.523 & 0.524 & 0.91 & 0.12 \\
MR-Median   & 0.168 & 0.170 & 0.173 & 0.62 & 0.59 & 0.168 & 0.249 & 0.251 & 0.57 & 0.56 & 0.166 & 0.307 & 0.309 & 0.53 & 0.55 \\
MR-Mode     & 0.116 & 0.191 & 0.209 & 0.90 & 0.10 & 0.077 & 0.369 & 0.389 & 0.90 & 0.08 & 0.022 & 0.928 & 0.944 & 0.88 & 0.09 \\
MR-Lasso    & 0.171 & 0.071 & 0.077 & 0.75 & 0.88 & 0.162 & 0.118 & 0.124 & 0.65 & 0.69 & 0.157 & 0.175 & 0.180 & 0.54 & 0.62 \\
MR-ALasso   & 0.172 & 0.067 & 0.073 & 0.78 & 0.89 & 0.162 & 0.100 & 0.107 & 0.70 & 0.70 & 0.154 & 0.148 & 0.155 & 0.59 & 0.60 \\
MR-ALasso-B & 0.174 & 0.062 & 0.068 & 0.96 & 0.66 & 0.164 & 0.093 & 0.100 & 0.95 & 0.28 & 0.157 & 0.135 & 0.141 & 0.96 & 0.15 \\

\midrule
\multicolumn{16}{l}{\textit{Scenario 2: Directional pleiotropy, InSIDE satisfied}} \\

MR-IVW      & 0.690 & 0.282 & 0.565 & 0.21 & 0.94 & 0.998 & 0.325 & 0.862 & 0.10 & 0.97 & 1.226 & 0.342 & 1.081 & 0.05 & 0.99 \\
MR-Robust   & 0.181 & 0.071 & 0.074 & 0.92 & 0.72 & 0.234 & 0.153 & 0.157 & 0.93 & 0.35 & 0.523 & 0.351 & 0.477 & 0.84 & 0.35 \\
MR-Egger    & 1.082 & 0.543 & 1.035 & 0.23 & 0.86 & 1.355 & 0.532 & 1.272 & 0.20 & 0.89 & 1.479 & 0.528 & 1.384 & 0.23 & 0.87 \\
MR-Median   & 0.397 & 0.215 & 0.292 & 0.46 & 0.93 & 0.686 & 0.357 & 0.603 & 0.20 & 0.96 & 0.950 & 0.414 & 0.857 & 0.09 & 0.98 \\
MR-Mode     & 0.134 & 0.175 & 0.187 & 0.92 & 0.12 & 0.152 & 0.493 & 0.495 & 0.93 & 0.11 & 0.122 & 1.038 & 1.040 & 0.93 & 0.13 \\
MR-Lasso    & 0.217 & 0.076 & 0.078 & 0.77 & 0.96 & 0.357 & 0.157 & 0.222 & 0.41 & 0.96 & 0.604 & 0.273 & 0.488 & 0.14 & 0.98 \\
MR-ALasso   & 0.184 & 0.068 & 0.070 & 0.81 & 0.92 & 0.218 & 0.107 & 0.109 & 0.73 & 0.85 & 0.318 & 0.166 & 0.204 & 0.50 & 0.87 \\
MR-ALasso-B & 0.187 & 0.063 & 0.065 & 0.96 & 0.71 & 0.225 & 0.098 & 0.101 & 0.97 & 0.43 & 0.334 & 0.156 & 0.206 & 0.95 & 0.37 \\

\midrule
\multicolumn{16}{l}{\textit{Scenario 3: Directional pleiotropy, InSIDE violated}} \\

MR-IVW      & 0.881 & 0.249 & 0.724 & 0.03 & 1.00 & 1.133 & 0.242 & 0.964 & 0.00 & 1.00 & 1.278 & 0.226 & 1.101 & 0.00 & 1.00 \\
MR-Robust   & 0.195 & 0.085 & 0.085 & 0.89 & 0.67 & 0.310 & 0.216 & 0.242 & 0.80 & 0.49 & 0.786 & 0.390 & 0.704 & 0.45 & 0.73 \\
MR-Egger    & 1.230 & 0.382 & 1.098 & 0.04 & 0.99 & 1.406 & 0.327 & 1.250 & 0.01 & 1.00 & 1.487 & 0.294 & 1.321 & 0.00 & 1.00 \\
MR-Median   & 0.743 & 0.325 & 0.633 & 0.08 & 0.99 & 1.057 & 0.321 & 0.915 & 0.01 & 1.00 & 1.221 & 0.297 & 1.063 & 0.00 & 1.00 \\
MR-Mode     & 0.237 & 0.681 & 0.682 & 0.79 & 0.29 & 0.338 & 0.576 & 0.592 & 0.77 & 0.31 & 0.424 & 0.703 & 0.737 & 0.77 & 0.35 \\
MR-Lasso    & 0.303 & 0.109 & 0.150 & 0.46 & 0.99 & 0.635 & 0.222 & 0.488 & 0.06 & 1.00 & 0.950 & 0.260 & 0.794 & 0.01 & 1.00 \\
MR-ALasso   & 0.203 & 0.076 & 0.076 & 0.74 & 0.94 & 0.290 & 0.125 & 0.154 & 0.48 & 0.94 & 0.479 & 0.180 & 0.332 & 0.15 & 0.98 \\
MR-ALasso-B & 0.206 & 0.069 & 0.069 & 0.97 & 0.72 & 0.296 & 0.111 & 0.147 & 0.94 & 0.56 & 0.490 & 0.159 & 0.330 & 0.78 & 0.62 \\

\midrule
\multicolumn{16}{l}{\textit{Scenario 4: Balanced pleiotropy, InSIDE violated}} \\

MR-IVW      & 0.222 & 0.280 & 0.281 & 0.65 & 0.48 & 0.256 & 0.317 & 0.322 & 0.75 & 0.39 & 0.283 & 0.323 & 0.333 & 0.81 & 0.32 \\
MR-Robust   & 0.172 & 0.072 & 0.077 & 0.90 & 0.67 & 0.165 & 0.153 & 0.157 & 0.90 & 0.26 & 0.201 & 0.307 & 0.307 & 0.88 & 0.18 \\
MR-Egger    & 0.273 & 0.522 & 0.527 & 0.60 & 0.45 & 0.362 & 0.502 & 0.528 & 0.76 & 0.34 & 0.443 & 0.485 & 0.543 & 0.83 & 0.27 \\
MR-Median   & 0.204 & 0.191 & 0.191 & 0.61 & 0.67 & 0.247 & 0.277 & 0.281 & 0.51 & 0.66 & 0.290 & 0.325 & 0.337 & 0.48 & 0.68 \\
MR-Mode     & 0.155 & 1.096 & 1.096 & 0.88 & 0.16 & 0.108 & 0.563 & 0.571 & 0.87 & 0.13 & 0.111 & 1.016 & 1.020 & 0.91 & 0.09 \\
MR-Lasso    & 0.176 & 0.072 & 0.076 & 0.76 & 0.89 & 0.179 & 0.122 & 0.124 & 0.63 & 0.72 & 0.198 & 0.181 & 0.181 & 0.54 & 0.70 \\
MR-ALasso   & 0.173 & 0.067 & 0.072 & 0.77 & 0.89 & 0.165 & 0.103 & 0.109 & 0.67 & 0.70 & 0.163 & 0.147 & 0.152 & 0.58 & 0.63 \\
MR-ALasso-B & 0.175 & 0.063 & 0.068 & 0.95 & 0.66 & 0.167 & 0.095 & 0.101 & 0.95 & 0.31 & 0.166 & 0.133 & 0.138 & 0.96 & 0.16 \\

\bottomrule
\end{tabular}
}
\begin{tablenotes}[flushleft]
\scriptsize
\item Note: $\hat{\theta}$ denotes the empirical mean of the estimated causal effect; SD is the empirical standard deviation; RMSE is the root mean squared error; CP is the empirical coverage probability of the nominal 95\% confidence interval; Power is the empirical probability of rejecting the null hypothesis $H_0:\theta=0$ when the true causal effect is $\theta=0.20$; and $\gamma_{\mathcal{A}} = 3\gamma_{\mathcal{V}}$ denotes the setting in which invalid instruments are stronger than valid instruments.
\end{tablenotes}

\label{tab:modelI_theta0.2_stronger_invalid}
\end{table}

\begin{table}[H]
\caption{Simulation results for different cases when the true causal effect is $\theta=0$ for Model-II.}
\centering
\resizebox{\textwidth}{!}{
\begin{tabular}{lcl ccccc ccccc}
\toprule
\multirow{2}{*}{Design} & \multirow{2}{*}{$p_I$} & \multirow{2}{*}{Method} 
& \multicolumn{5}{c}{$n=10000$} & \multicolumn{5}{c}{$n=20000$} \\
\cmidrule(lr){4-8} \cmidrule(lr){9-13}
& & & $\hat{\theta}$ & SD & RMSE & $\bar{s}_{\text{sel}}$ & All invalid
& $\hat{\theta}$ & SD & RMSE & $\bar{s}_{\text{sel}}$ & All invalid \\
\midrule

\multirow{15}{*}{$\gamma_{\mathcal A}=\gamma_{\mathcal V}$}
& \multirow{5}{*}{15\%}
& MR-IVW        & 0.281 & 0.029 & 0.283 &  --   &   --   & 0.285 & 0.026 & 0.286 &  --   &  --   \\
&               & Oracle IVW     & 0.000 & 0.017 & 0.017 & 15 & 1 & 0.000 & 0.012 & 0.012 & 15 & 1 \\
&               & MR-Lasso      & 0.003 & 0.018 & 0.018 & 16.1 & 1 & 0.002 & 0.013 & 0.013 & 16.1 & 1 \\
&               & MR-ALasso     & 0.001 & 0.017 & 0.017 & 15.8 & 1 & 0.001 & 0.012 & 0.012 & 15.1 & 1 \\
&               & MR-ALasso-B   & 0.002 & 0.017 & 0.017 & 15.9 & 1 & 0.001 & 0.012 & 0.012 & 15.9 & 1 \\
\cmidrule(lr){2-13}

& \multirow{5}{*}{30\%}
& MR-IVW        & 0.563 & 0.037 & 0.564 &  --   &   --   & 0.570 & 0.034 & 0.571 &  --   &  --   \\
&               & Oracle IVW     & 0.000 & 0.018 & 0.018 & 30 & 1 & 0.000 & 0.013 & 0.013 & 30 & 1 \\
&               & MR-Lasso      & 0.005 & 0.020 & 0.021 & 30.9 & 1 & 0.003 & 0.014 & 0.015 & 30.8 & 1 \\
&               & MR-ALasso     & 0.002 & 0.019 & 0.019 & 30.6 & 1 & 0.001 & 0.013 & 0.014 & 30.1 & 1 \\
&               & MR-ALasso-B   & 0.005 & 0.019 & 0.019 & 30.7 & 1 & 0.003 & 0.013 & 0.014 & 30.7 & 1 \\
\cmidrule(lr){2-13}

& \multirow{5}{*}{45\%}
& MR-IVW        & 0.844 & 0.043 & 0.845 &  --   &   --   & 0.854 & 0.039 & 0.855 &  --   &  --   \\
&               & Oracle IVW     & 0.000 & 0.021 & 0.021 & 45 & 1 & 0.000 & 0.015 & 0.015 & 45 & 1 \\
&               & MR-Lasso      & 0.077 & 0.043 & 0.088 & 62.6 & 1 & 0.037 & 0.025 & 0.045 & 58.9 & 1 \\
&               & MR-ALasso     & 0.002 & 0.022 & 0.022 & 45.4 & 1 & 0.002 & 0.015 & 0.016 & 45.1 & 1 \\
&               & MR-ALasso-B   & 0.233 & 0.052 & 0.238 & 45.8 & 1 & 0.218 & 0.045 & 0.222 & 45.7 & 1 \\

\midrule

\multirow{15}{*}{$\gamma_{\mathcal A}=3\gamma_{\mathcal V}$}
& \multirow{5}{*}{15\%}
& MR-IVW        & 0.389 & 0.018 & 0.389 &  --   &  --   & 0.391 & 0.016 & 0.392 &  --   &  --   \\
&               & Oracle IVW     & 0.000 & 0.017 & 0.017 & 15 & 1 & 0.000 & 0.012 & 0.012 & 15 & 1 \\
&               & MR-Lasso      & 0.019 & 0.022 & 0.029 & 19.5 & 0.998 & 0.006 & 0.014 & 0.015 & 16.9 & 1 \\
&               & MR-ALasso     & 0.001 & 0.017 & 0.017 & 15.7 & 1 & 0.001 & 0.012 & 0.012 & 15.1 & 1 \\
&               & MR-ALasso-B   & 0.003 & 0.017 & 0.017 & 15.9 & 1 & 0.002 & 0.012 & 0.012 & 15.9 & 1 \\
\cmidrule(lr){2-13}

& \multirow{5}{*}{30\%}
& MR-IVW        & 0.506 & 0.016 & 0.507 &  --   &  --   & 0.508 & 0.015 & 0.508 &  --   &  --   \\
&               & Oracle IVW     & 0.000 & 0.018 & 0.018 & 30 & 1 & 0.000 & 0.013 & 0.013 & 30 & 1 \\
&               & MR-Lasso      & 0.509 & 0.035 & 0.510 & 79.6 & 0 & 0.520 & 0.033 & 0.521 & 86.6 & 0 \\
&               & MR-ALasso     & 0.003 & 0.019 & 0.020 & 30.7 & 1 & 0.002 & 0.013 & 0.014 & 30.1 & 1 \\
&               & MR-ALasso-B   & 0.008 & 0.020 & 0.022 & 31.0 & 1 & 0.004 & 0.014 & 0.014 & 30.7 & 1 \\
\cmidrule(lr){2-13}

& \multirow{5}{*}{45\%}
& MR-IVW        & 0.562 & 0.016 & 0.562 &  --   &  --   & 0.564 & 0.015 & 0.564 &  --   &  --   \\
&               & Oracle IVW     & 0.000 & 0.021 & 0.021 & 45 & 1 & 0.000 & 0.015 & 0.015 & 45 & 1 \\
&               & MR-Lasso      & 0.580 & 0.025 & 0.580 & 74.0 & 0 & 0.584 & 0.024 & 0.585 & 80.4 & 0 \\
&               & MR-ALasso     & 0.096 & 0.092 & 0.133 & 60.4 & 0.921 & 0.013 & 0.021 & 0.025 & 47.9 & 1 \\
&               & MR-ALasso-B   & 0.193 & 0.055 & 0.200 & 59.2 & 0.919 & 0.115 & 0.031 & 0.119 & 48.2 & 1 \\

\bottomrule
\end{tabular}
}
\begin{tablenotes}[flushleft]
\scriptsize
\item Note: Oracle IVW is the IVW estimator computed using the true valid set of instruments. $\hat{\theta}$ denotes the empirical mean of the estimated causal effect; SD is the empirical standard deviation; RMSE is the root mean squared error; $\bar{s}_{\mathrm{sel}}$ is the average number of instruments identified as invalid; and ``All invalid'' is the empirical proportion of replicates in which the selected invalid set contains all truly invalid instruments, i.e., $\mathcal A\subseteq\widehat{\mathcal A}_n$.
\end{tablenotes}
\label{tab:modelII_theta0.0}
\end{table}

\begin{table}[H]
\caption{Simulation results for different cases when the true causal effect is $\theta=0.20$ for Model-II.}
\centering
\resizebox{\textwidth}{!}{
\begin{tabular}{lcl ccccc ccccc}
\toprule
\multirow{2}{*}{Design} & \multirow{2}{*}{$p_I$} & \multirow{2}{*}{Method} 
& \multicolumn{5}{c}{$n=10000$} & \multicolumn{5}{c}{$n=20000$} \\
\cmidrule(lr){4-8} \cmidrule(lr){9-13}
& & & $\hat{\theta}$ & SD & RMSE & $\bar{s}_{\text{sel}}$ & All invalid
& $\hat{\theta}$ & SD & RMSE & $\bar{s}_{\text{sel}}$ & All invalid \\
\midrule

\multirow{15}{*}{$\gamma_{\mathcal A}=\gamma_{\mathcal V}$}
& \multirow{5}{*}{15\%}
& MR-IVW        & 0.476 & 0.030 & 0.278 &  --   &   --   & 0.482 & 0.027 & 0.283 &  --   & -- \\
&               & Oracle IVW    & 0.195 & 0.017 & 0.017 & 15 & 1 & 0.198 & 0.012 & 0.012 & 15 & 1 \\
&               & MR-Lasso      & 0.199 & 0.018 & 0.018 & 16.6 & 1 & 0.201 & 0.013 & 0.013 & 16.6 & 1 \\
&               & MR-ALasso     & 0.197 & 0.018 & 0.018 & 16.0 & 1 & 0.198 & 0.012 & 0.012 & 15.2 & 1 \\
&               & MR-ALasso-B   & 0.198 & 0.018 & 0.018 & 16.2 & 1 & 0.199 & 0.012 & 0.012 & 16.1 & 1 \\
\cmidrule(lr){2-13}

& \multirow{5}{*}{30\%}
& MR-IVW        & 0.758 & 0.038 & 0.559 &  --   &   --   & 0.767 & 0.035 & 0.568 &  --   & -- \\
&               & Oracle IVW    & 0.195 & 0.018 & 0.019 & 30 & 1 & 0.198 & 0.013 & 0.013 & 30 & 1 \\
&               & MR-Lasso      & 0.202 & 0.022 & 0.022 & 31.3 & 1 & 0.202 & 0.016 & 0.016 & 31.1 & 1 \\
&               & MR-ALasso     & 0.198 & 0.020 & 0.020 & 30.8 & 1 & 0.199 & 0.014 & 0.014 & 30.9 & 1 \\
&               & MR-ALasso-B   & 0.202 & 0.019 & 0.019 & 31.0 & 1 & 0.202 & 0.014 & 0.014 & 31.0 & 1 \\
\cmidrule(lr){2-13}

& \multirow{5}{*}{45\%}
& MR-IVW        & 1.040 & 0.043 & 0.841 &  --   &   --   & 1.052 & 0.040 & 0.853 &  --   & -- \\
&               & Oracle IVW    & 0.195 & 0.021 & 0.021 & 45 & 1 & 0.198 & 0.015 & 0.015 & 45 & 1 \\
&               & MR-Lasso      & 0.285 & 0.049 & 0.098 & 63.8 & 1 & 0.241 & 0.028 & 0.049 & 59.6 & 1 \\
&               & MR-ALasso     & 0.199 & 0.023 & 0.023 & 45.5 & 1 & 0.200 & 0.016 & 0.016 & 45.2 & 1 \\
&               & MR-ALasso-B   & 0.433 & 0.055 & 0.240 & 46.1 & 1 & 0.416 & 0.046 & 0.221 & 45.9 & 1 \\

\midrule

\multirow{15}{*}{$\gamma_{\mathcal A}=3\gamma_{\mathcal V}$}
& \multirow{5}{*}{15\%}
& MR-IVW        & 0.587 & 0.019 & 0.387 &  --   & -- & 0.590 & 0.016 & 0.390 &  --   & -- \\
&               & Oracle IVW    & 0.195 & 0.017 & 0.017 & 15 & 1 & 0.198 & 0.012 & 0.012 & 15 & 1 \\
&               & MR-Lasso      & 0.218 & 0.023 & 0.029 & 20.2 & 0.998 & 0.206 & 0.015 & 0.016 & 17.4 & 1 \\
&               & MR-ALasso     & 0.197 & 0.018 & 0.018 & 16.0 & 1 & 0.199 & 0.012 & 0.013 & 15.1 & 1 \\
&               & MR-ALasso-B   & 0.199 & 0.018 & 0.018 & 16.2 & 1 & 0.200 & 0.013 & 0.012 & 16.2 & 1 \\
\cmidrule(lr){2-13}

& \multirow{5}{*}{30\%}
& MR-IVW        & 0.705 & 0.017 & 0.505 &  --   & -- & 0.707 & 0.015 & 0.508  &  --   & -- \\ 
&               & Oracle IVW    & 0.195 & 0.018 & 0.019 & 30 & 1 & 0.198 & 0.013 & 0.013 & 30 & 1 \\
&               & MR-Lasso      & 0.706 & 0.035 & 0.507 & 78.5 & 0.001 & 0.718 & 0.034 & 0.519 & 86.2 & 0 \\
&               & MR-ALasso     & 0.200 & 0.021 & 0.021 & 30.5 & 1 & 0.200 & 0.014 & 0.014 & 30.1 & 1 \\
&               & MR-ALasso-B   & 0.207 & 0.021 & 0.022 & 31.3 & 1 & 0.203 & 0.014 & 0.014 & 31.0 & 1 \\
\cmidrule(lr){2-13}

& \multirow{5}{*}{45\%}
& MR-IVW        & 0.761 & 0.016 & 0.561 &  --   & -- & 0.763 & 0.015 & 0.563 &  --   & -- \\
&               & Oracle IVW    & 0.195 & 0.021 & 0.021 & 45 & 1 & 0.198 & 0.015 & 0.015 & 45 & 1 \\
&               & MR-Lasso      & 0.778 & 0.025 & 0.578 & 73.6 & 0 & 0.783 & 0.024 & 0.583 & 80.2 & 0 \\
&               & MR-ALasso     & 0.331 & 0.121 & 0.178 & 63.7 & 0.833 & 0.217 & 0.026 & 0.030 & 49.1 & 1 \\
&               & MR-ALasso-B   & 0.411 & 0.061 & 0.220 & 62.2 & 0.845 & 0.322 & 0.035 & 0.127 & 49.3 & 1 \\

\bottomrule
\end{tabular}
}
\begin{tablenotes}[flushleft]
\scriptsize
\item Note: Oracle IVW is the IVW estimator computed using the true valid set of instruments. $\hat{\theta}$ denotes the empirical mean of the estimated causal effect; SD is the empirical standard deviation; RMSE is the root mean squared error; $\bar{s}_{\mathrm{sel}}$ is the average number of instruments identified as invalid; and ``All invalid'' is the empirical proportion of replicates in which the selected invalid set contains all truly invalid instruments, i.e., $\mathcal A\subseteq\widehat{\mathcal A}_n$.
\end{tablenotes}
\label{tab:modelII_theta0.2}
\end{table}

\end{document}